\documentclass[a4paper,11pt]{article}
\usepackage{amsfonts}
\usepackage{amsmath}
\usepackage{amsthm}
\usepackage{amssymb}
\usepackage{stmaryrd}
\usepackage{relsize}
\usepackage{graphicx, subfigure}
\usepackage[hyphens]{url}
\usepackage{hyperref}
\usepackage{xspace}

\newtheorem{theorem}{Theorem}[section]
\newtheorem{cor}[theorem]{Corollary}
\newtheorem{proposition}[theorem]{Proposition}
\newtheorem{lemma}[theorem]{Lemma}
\newtheorem{observation}[theorem]{Observation}

\newcommand{\R}{\mathbb{R}}
\newcommand{\complexityclass}[1]{\textbf{#1}}
\newcommand{\computproblem}[1]{\textsc{#1}}
\newcommand{\NP}{\complexityclass{NP}\xspace}

\newcommand{\SET}{\computproblem{Set Cover}\xspace}

\makeatletter
\def\moverlay{\mathpalette\mov@rlay}
\def\mov@rlay#1#2{\leavevmode\vtop{
   \baselineskip\z@skip \lineskiplimit-\maxdimen
   \ialign{\hfil$\m@th#1##$\hfil\cr#2\crcr}}}
\newcommand{\charfusion}[3][\mathord]{
    #1{\ifx#1\mathop\vphantom{#2}\fi
        \mathpalette\mov@rlay{#2\cr#3}
      }
    \ifx#1\mathop\expandafter\displaylimits\fi}
\makeatother

\newcommand{\ls}[1]{\charfusion[\mathop]{\mathlarger\bigsqcup}{#1}}
\newcommand{\rs}[1]{\charfusion[\mathop]{\mathlarger\bigsqcap}{#1}}
\newcommand{\lrs}[1]{\charfusion[\mathop]{\mathlarger\bigbox}{#1}}

\title{Face-Guarding Polyhedra}

\author{Giovanni Viglietta\thanks{School of Electrical Engineering and Computer Science, University of Ottawa, Ottawa ON, Canada, {\tt viglietta@gmail.com}.}}

\index{Viglietta, Giovanni}

\begin{document}
\thispagestyle{empty}
\maketitle

\begin{abstract}
We study the Art Gallery Problem for face guards in polyhedral environments. The problem can be informally stated as: \emph{how many (not necessarily convex) windows should we place on the external walls of a dark building, in order to completely illuminate its interior?}

We consider both \emph{closed} and \emph{open} face guards (i.e., faces with or without their boundary), and we study several classes of polyhedra, including \emph{orthogonal} polyhedra, \emph{4-oriented} polyhedra, and \emph{2-reflex orthostacks}.

We give upper and lower bounds on the minimum number of faces required to guard the interior of a given polyhedron in each of these classes, in terms of the total number of its faces, $f$. In several cases our bounds are tight: $\lfloor f/6\rfloor$ \emph{open} face guards for orthogonal polyhedra and 2-reflex orthostacks, and $\lfloor f/4\rfloor$ \emph{open} face guards for $4$-oriented polyhedra. Additionally, for \emph{closed} face guards in 2-reflex orthostacks, we give a lower bound of $\lfloor (f+3)/9\rfloor$ and an upper bound of $\lfloor (f+1)/7\rfloor$.

Then we show that it is \NP-hard to approximate the minimum number of (closed or open) face guards within a factor of $\Omega(\log f)$, even for polyhedra that are orthogonal and simply connected. We also obtain the same hardness results for \emph{polyhedral terrains}.

Along the way we discuss some applications, arguing that face guards are \emph{not} a reasonable model for guards \emph{patrolling} on the surface of a polyhedron.
\end{abstract}

\section{Introduction}\label{s1}

\paragraph{Previous work.}

Art Gallery Problems have been studied in computational geometry for decades: given an \emph{enclosure}, place a (preferably small) set of \emph{guards} such that every location in the enclosure is seen by some guard. Most of the early research on the Art Gallery Problem focused on guarding 2-dimensional polygons with either point guards or segment guards~\cite{art,shermer,urrutia2000}.

Gradually, some of the attention started shifting to 3-dimensional settings, as well. Several authors have considered edge guards in 3-dimensional polyhedra, either in relation to the classical Art Gallery Problem or to its variations~\cite{viglietta4,edgenew,wireless,viglietta2,thesis}.

Recently, Souvaine et al.~\cite{faceguards} introduced the model with \emph{face guards} in 3-dimensional polyhedra. Ideally, each guard is free to roam over an entire face of a polyhedron, including the face's boundary. Let $g(\mathcal P)$ be the minimum number of face guards needed for a polyhedron $\mathcal P$, and let $g(f)$ be the maximum of $g(\mathcal P)$ over all polyhedra $\mathcal P$ with exactly $f$ faces. For general polyhedra, Souvaine et al.\ showed that $\lfloor f/5\rfloor \leqslant g(f) \leqslant \lfloor f/2\rfloor$ and, for the special case of orthogonal polyhedra (i.e., polyhedra whose faces are orthogonal to the coordinate axes), they showed that $\lfloor f/7\rfloor \leqslant g(f) \leqslant \lfloor f/6\rfloor$. They also suggested several open problems, such as studying \emph{open} face guards (i.e., face guards whose boundary is omitted), and the computational complexity of minimizing the number of face guards.

Subsequently, face guards have been studied to some extent also in the case of polyhedral terrains. In~\cite{terrain2,terrain3} a tight bound is obtained, and in~\cite{terrain1} it is proven that minimizing face guards in triangulated terrains is \NP-hard. However, since these results apply to terrains, they have no direct implications on the problem of face-guarding polyhedral enclosures.

\paragraph{Our contribution.}

In this paper we solve some of the problems left open in~\cite{faceguards}, and we also expand our research in some new directions. A preliminary version of this paper has appeared at CCCG 2013~\cite{mycccg}.

In Section~\ref{s2} we discuss the face guard model, arguing that a face guard fails to meaningfully represent a guard ``patrolling'' on a face of a polyhedron. Essentially, there are cases in which the path that such a patrolling guard ought to follow is so complex (in terms of the number of turns, if it is a polygonal chain) that a much simpler path, striving from the face, would guard not only the region visible from that face, but the entire polyhedron.
However, face guards are still a good model for illumination-related problems, such as placing (possibly non-convex) windows in a dark building.

In Section~\ref{s3} we obtain some new bounds on $g(f)$, for both closed and open face guards. First we generalize the upper bounds given in~\cite{faceguards} by showing that, for $c$-oriented polyhedra (i.e., whose faces have $c$ distinct orientations), $g(f)\leqslant \lfloor f/2 - f/c\rfloor$. We also provide some new lower bound constructions, which meet our upper bounds in two notable cases: orthogonal polyhedra with open face guards ($g(f)=\lfloor f/6\rfloor$), and $4$-oriented polyhedra with open face guards ($g(f)=\lfloor f/4\rfloor$). Then we go on to study a special class of orthogonal polyhedra, namely \emph{2-reflex orthostacks}.

The following table summarizes our new results, as well as those that were already known. Each entry contains a lower and an upper bound on $g(f)$, or a single tight bound. When applicable, a reference is given to the paper in which each result was first obtained. Observe that, for open face guards in triangulated terrains, $f$ guards are easily seen to be necessary in the worst case. Indeed, if the terrain is a convex ``dome'' (i.e., if no edges are reflex), then every face requires an open face guard. In the case of closed face guards in triangulated terrains, we remark that the bound given in~\cite{terrain3} is expressed in terms of the number of vertices. Therefore we rewrote it in terms of $f$, using Euler's formula.

\begin{center}
\begin{tabular}{c|c|c}
 & \textbf{Open face guards} & \textbf{Closed face guards}\\
\hline
\textbf{2-reflex orthostacks} & $g(f)=\lfloor f/6\rfloor$ & $\lfloor (f+3)/9\rfloor\leqslant g(f)\leqslant\lfloor (f+1)/7\rfloor$ \\
\textbf{Orthogonal polyhedra} & $g(f)=\lfloor f/6\rfloor$ & $\lfloor f/7\rfloor\leqslant_{\cite{faceguards}} g(f)\leqslant_{\cite{faceguards}}\lfloor f/6\rfloor$\\
\textbf{4-oriented polyhedra} & $g(f)=\lfloor f/4\rfloor$ & $\lfloor f/5\rfloor\leqslant g(f)\leqslant\lfloor f/4\rfloor$ \\
\textbf{General polyhedra} & $\lfloor f/4\rfloor\leqslant g(f)\leqslant\lfloor f/2\rfloor-1$ & $\lfloor f/5\rfloor\leqslant_{\cite{faceguards}} g(f)\leqslant\lfloor f/2\rfloor-1$ \\
\textbf{Triangulated terrains} & $g(f)=f$ & $g(f)=_{\cite{terrain3}}\lfloor (f+3)/6\rfloor$ \\
\end{tabular}
\end{center}

In Section~\ref{s4} we provide an approximation-preserving reduction from \SET to the problem of minimizing the number of (closed or open) face guards in simply connected orthogonal polyhedra. It follows that the minimum number of face guards is \NP-hard to approximate within a factor of $\Omega(\log f)$. We also obtain the same result for (non-triangulated) terrains. This adds to the result of~\cite{terrain1}, which states that minimizing closed face guards is \NP-hard in triangulated terrains. We also briefly discuss the membership in \NP of the minimization problem, pointing out some difficulties in applying previously known techniques.

We leave as an open problem the task to tighten all the bounds in the table above, as well as to prove or disprove that minimizing face guards is in \NP. We conjecture that all the lower bounds are tight, and that the minimization problem does belong to \NP.

\section{Model and motivations}\label{s2}

\paragraph{Definitions.}

A \emph{polyhedron} is a connected subset of $\R^3$, union of finitely many closed tetrahedra embedded in $\R^3$, whose boundary is a (possibly non-connected) orientable 2-manifold. Since a polyhedron's boundary is piecewise linear, the notion of \emph{face} of a polyhedron is well defined as a maximal planar subset of its boundary with connected and non-empty relative interior. Thus a face is a plane polygon, possibly with holes, and possibly with some degeneracies, such as hole boundaries touching each other at a single vertex. Any vertex of a face is also considered a \emph{vertex} of the polyhedron. \emph{Edges} are defined as minimal non-degenerate straight line segments shared by two distinct faces and connecting two vertices of the polyhedron. Since a polyhedron's boundary is an orientable 2-manifold, the relative interior of an edge lies on the boundary of exactly two faces, thus determining an internal dihedral angle (with respect to the polyhedron). An edge is \emph{reflex} if its internal dihedral angle is reflex, i.e., strictly greater than $180^\circ$.

Given a polyhedron, we say that a point $x$ is \emph{visible} to a point $y$ if no point in the straight line segment $xy$ lies in the exterior of the polyhedron. For any point $x$, we denote by $\mathcal V(x)$ the \emph{visible region} of $x$, i.e., the set of points that are visible to $x$. In general, for any set $S\subset \mathbb R^3$, we let $\mathcal V(S)=\bigcup_{x\in S}\mathcal V(x)$. A set is said to \emph{guard} a polyhedron if its visible region coincides with the entire polyhedron (including its boundary). The Art Gallery Problem for face guards in polyhedra consists in finding a (preferably small) set of faces whose union guards a given polyhedron. If such faces include their relative boundary, they are called \emph{closed} face guards; if their boundary is omitted, they are called \emph{open} face guards.

A polyhedron is \emph{$c$-oriented} if there exist $c$ unit vectors such that each face is orthogonal to one of the vectors. If these unit vectors form an orthonormal basis of $\mathbb R^3$, the polyhedron is said to be \emph{orthogonal}. Hence, a cube is orthogonal, a tetrahedron and a regular octahedron are both 4-oriented, etc.

We will refer informally to the $z$ axis as the \emph{vertical} axis. Specifically, the positive $z$ direction will be \emph{up}, and the opposite direction will be \emph{down}. Hence, a direction parallel to the $xy$ plane will be said to be \emph{horizontal}. The positive $x$ direction will be \emph{right}, the negative $x$ direction will be \emph{left}, and so on.

\paragraph{Motivations.}

There is a straightforward analogy between \emph{guarding} problems and \emph{illumination} problems: placing guards in a polyhedron corresponds to placing light sources in a dark building, in order to illuminate it completely. For instance, a point guard would model a \emph{light bulb} and a segment guard could be a \emph{fluorescent tube}. Because face guards are 2-dimensional and lie on the boundary of the polyhedron, we may think  of them as \emph{windows}. A window may have any shape, but should be flat, and hence it should lie on a single face. It follows that, if our purpose is to illuminate as big a region as possible, we may assume without loss of generality that a window always coincides with some face.

Face guards were introduced in~\cite{faceguards} to represent guards \emph{roaming} over a face. This is in accordance with the traditional usage of segment guards as a model for guards that \emph{patrol} on a line~\cite{art}. While this is perfectly sound in the case of segment guards, face guards pose additional problems, as explained next.

\begin{figure}[h]
\centering
\subfigure[]{\label{fig1:a}\includegraphics[height=.35\linewidth]{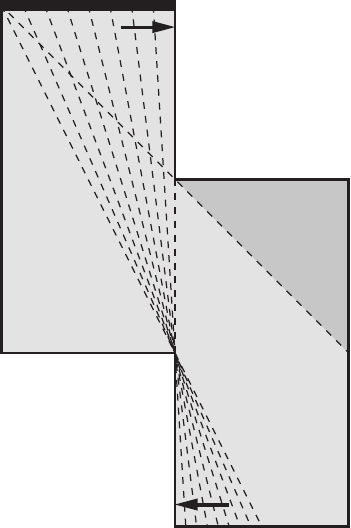}}\qquad\qquad
\subfigure[]{\label{fig1:b}\includegraphics[height=.35\linewidth]{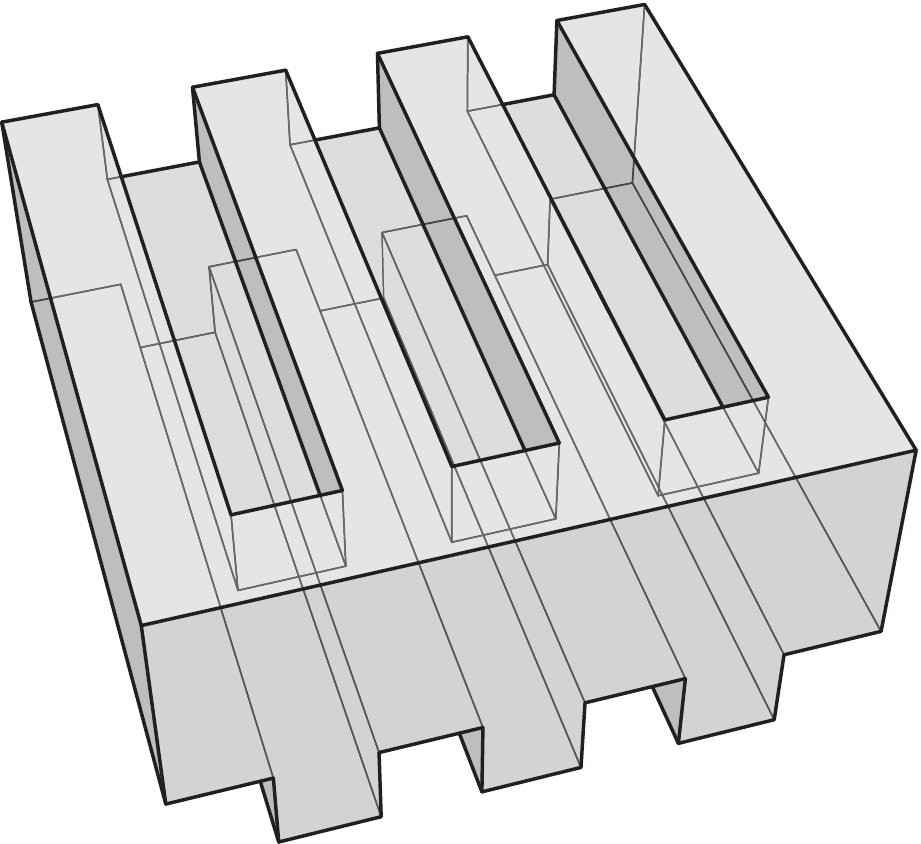}}
\caption{Constructing the polyhedron in Figure~\ref{fig2}}
\label{fig1}
\end{figure}

We begin by observing that, even in 2-dimensional polygons, there may be edge guards that cannot be locally ``replaced'' by finitely many point guards. Figure~\ref{fig1:a} shows an example: if a subset $G$ of the top edge $\ell$ is such that $\mathcal V(G)=\mathcal V(\ell)$, then the right endpoint of $\ell$ must be a limit point of $G$.

We can exploit this fact to construct the class of polyhedra sketched in Figure~\ref{fig2}. First we cut long parallel \emph{dents} on opposite faces of a cuboid, as in Figure~\ref{fig1:b}, in such a way that the resulting polyhedron looks like an extruded ``iteration'' of the polygon in Figure~\ref{fig1:a}. Then we stab this construction with a row of \emph{girders} running orthogonally with respect to the dents, as Figure~\ref{fig2:a} illustrates.

\begin{figure}[h]
\centering
\subfigure[]{\label{fig2:a}\includegraphics[height=.35\linewidth]{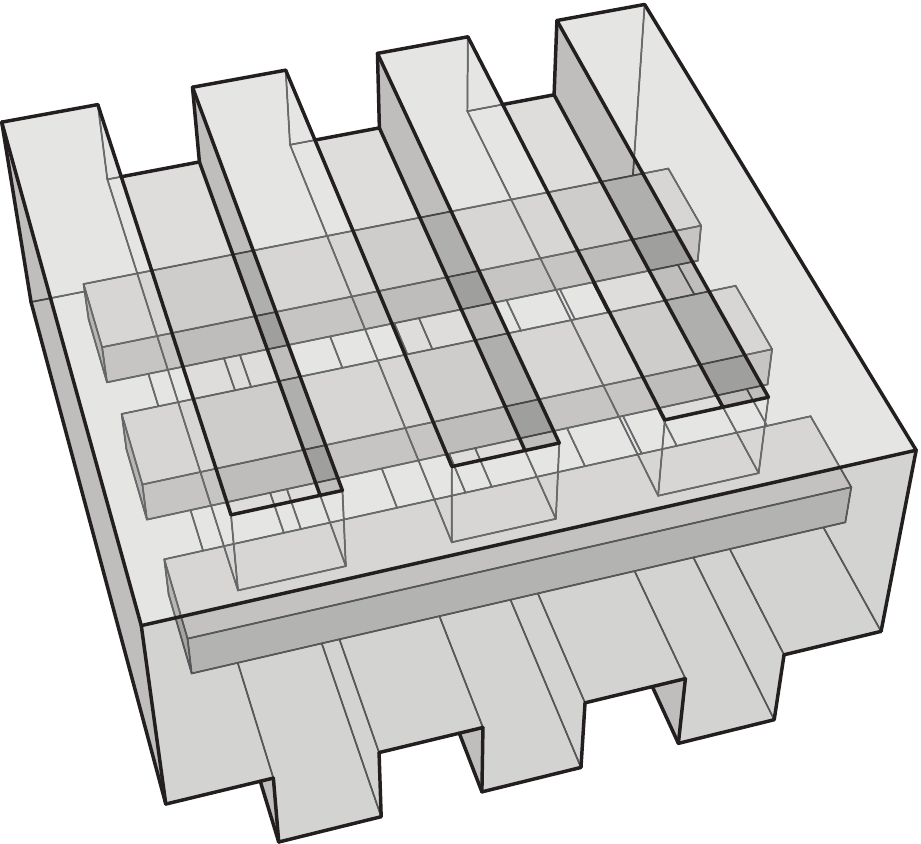}}\qquad\qquad
\subfigure[]{\label{fig2:b}\includegraphics[height=.35\linewidth]{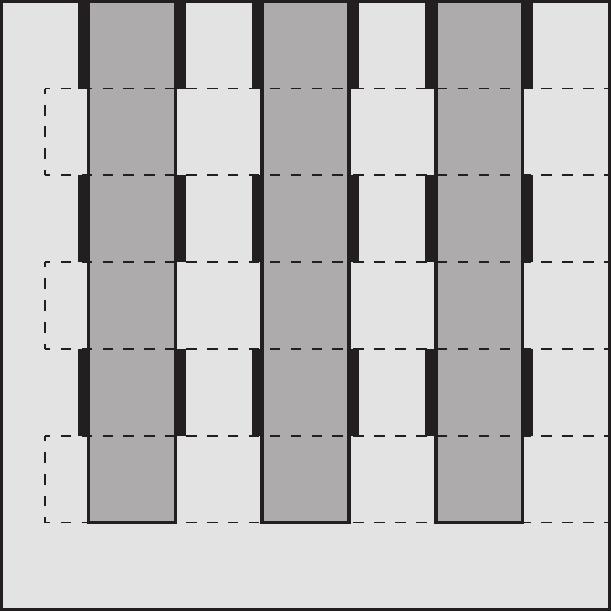}}
\caption{A guard patrolling on the top face must follow a path of quadratic complexity}
\label{fig2}
\end{figure}

Suppose that a guard has to patrol the top face of this construction, eventually seeing every point that is visible from that face. The situation is represented in Figure~\ref{fig2:b}, where the light-shaded region is the top face, and the dashed lines mark the underlying girders. By the above observation and by the presence of the girders, each thick vertical segment must be approached by the patrolling guard from the interior of the face.

Suppose that the polyhedron has $n$ dents and $n$ girders. Then, the number of its vertices, edges, or faces is~$\Theta(n)$. Now, if the guard moves along a polygonal chain lying on the top face, such a chain must have at least a vertex on each thick segment, which amounts to~$\Omega(n^2)$ vertices. Similarly, if the face guard has to be substituted with segment guards lying on it, quadratically many guards are needed.

On the other hand, it is easy to show that a path of linear complexity is sufficient to guard any polyhedron, provided that its boundary is connected: triangulate every face (thus adding linearly many new ``edges'') and traverse the resulting 1-skeleton in depth-first order starting from any vertex, thus covering all edges. Because the set of edges is a guarding set for any polyhedron~\cite[Observation~3.10]{thesis}, the claim follows. In general, if the boundary is not connected (i.e., the polyhedron has some ``cavities''), then we may repeat the same construction for every connected component.

This defeats the purpose of having faces model guards patrolling on segments, as it makes little sense for a face of ``unit weight'' to represent quadratically many guards. Analogously, a roaming guard represented by a face may have to follow a path that is overly complex compared to the guarding problem's optimal solution.

Even if we are allowed to replace a face guard with guards patrolling any segment in the polyhedron (i.e, not necessarily constrained to live on that face), a linear number of them may be required. Indeed, consider a cuboid with very small height, and arrange $n$ thin and long \emph{chimneys} on its top, in such a way that no straight line intersects more than two chimneys. The complexity of the polyhedron is $\Theta(n)$, and a face guard lying on the bottom face must be replaced by $\Omega(n)$ segment guards. On the other hand, we know that a linear number of segment guards is enough not only to ``dominate'' a single face, but to entirely guard any polyhedron.

Summarizing, a face guard appropriately models an entity that is naturally constrained to live on a single face, like a flat window, and unlike a team of patrolling guards. In the case of a single roaming guard, the model is insensitive to the complexity of the guard's path.

\section{Bounds on face guard numbers}\label{s3}

\subsection{$c$-oriented polyhedra}

Here we parameterize polyhedra according to the orientations of their faces, and we give upper and lower bounds on the number of face guards required to guard them.

\paragraph{Upper bounds.}

By generalizing the approach used in~\cite[Lemmas~2.1,~3.1]{faceguards}, we provide an upper bound on face guard numbers, which becomes tight for open face guards in orthogonal polyhedra and open face guards in 4-oriented polyhedra. We emphasize that our upper bound holds for both closed and open face guards, and for polyhedra of any genus and number of cavities.

\begin{theorem}\label{t3:face}
Any $c$-oriented polyhedron with $f$ faces is guardable by
$$\left\lfloor\frac f 2 - \frac f c\right\rfloor$$
closed or open face guards.
\end{theorem}
\begin{proof}
Let $\mathcal P$ be a polyhedron whose faces are orthogonal to $c\geqslant 3$ distinct vectors. Let $f_i$ be the number of faces orthogonal to the $i$-th vector $v_i$. We may assume that $i<j$ implies $f_i \geqslant f_j$. Then,
$$f_1+f_2 \geqslant \frac {2f} c.$$
Let us stipulate that the direction of the cross product $v_1 \times v_2$ is \emph{vertical}. Thus, there are at most
$$f-\frac {2f} c$$
non-vertical faces. Some of these are facing up, the others are facing down. Without loss of generality, at most half of them are facing down, and we assign a face guard to each of them. Therefore, at most
$$\left\lfloor\frac f 2 - \frac f c\right\rfloor$$
face guards have been assigned.

Let $x$ be any point in $\mathcal P$. If $x$ belongs to a face with a guard, $x$ is guarded. Otherwise, consider an infinite circular cone $\mathcal C(x)$ with apex $x$ and axis directed upward. Let $\mathcal G$ be  the intersection of $\mathcal V(x)$, $\mathcal C(x)$, and the boundary of $\mathcal P$. Intuitively, $\mathcal G$ is the part of the boundary of the polyhedron that would be illuminated by a spotlight placed at $x$ and pointed upward. We will show that this area contains points belonging to some guard, which make $x$ guarded.

If the aperture of $\mathcal C(x)$ is small enough, the relative interior of $\mathcal G$ belongs entirely to faces containing guards and to at most two vertical faces containing $x$. Because these vertical faces obstruct at most one dihedral angle from $x$'s view, the portion of $\mathcal G$ not belonging to them has non-empty interior. If we remove from this portion the (finitely many) edges of $\mathcal P$, we still have a non-empty region. By construction, this region belongs to the interiors of faces containing a guard; hence $x$ is guarded.
\end{proof}

Our guarding strategy becomes less efficient as $c$ grows. In general, if no two faces are parallel (i.e., $c=f$), we get an upper bound of $\left \lfloor f/2\right\rfloor - 1$ face guards, which improves on the one in~\cite{faceguards} by just one unit.

\paragraph{Lower bounds.}

In~\cite{faceguards}, Souvaine et al.\ construct a class of orthogonal polyhedra with $f$ faces that need $\lfloor f/7\rfloor$ \emph{closed} face guards. In Figure~\ref{fig3} we give an alternative construction, with the additional property of having a 3-regular 1-skeleton and having no vertical reflex edges. Indeed, each small L-shaped polyhedron that is attached to the big cuboid adds seven faces to the construction, of which at least one must be selected.

\begin{figure}[h]
\centering
\includegraphics[width=.5\linewidth]{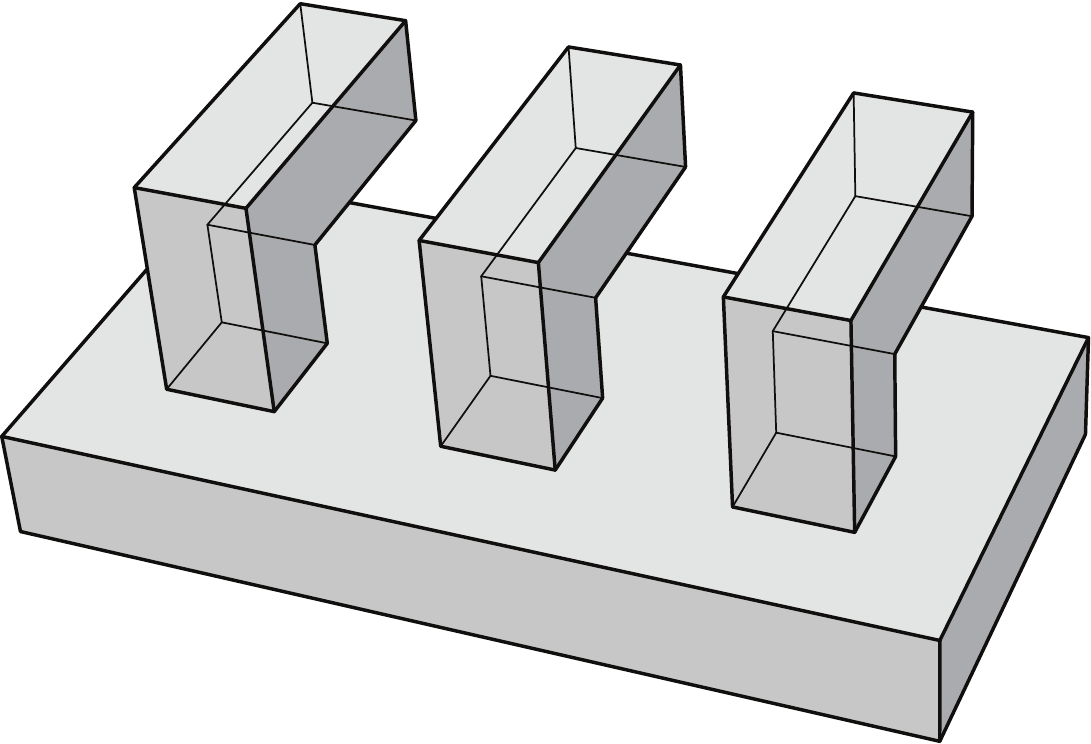}
\caption{Orthogonal polyhedron that needs $\lfloor f/7\rfloor$ closed face guards}
\label{fig3}
\end{figure}

For \emph{open} face guards, we have a different construction, shown in Figure~\ref{fig4}. There are six faces for each of the large flat cuboids, and no open face can guard the center of two different flat cuboids. Therefore, one guard is needed for each of the flat cuboids, and this amounts to $\lfloor f/6\rfloor$ open face guards.

\begin{figure}[h]
\centering{\includegraphics[width=.5\linewidth]{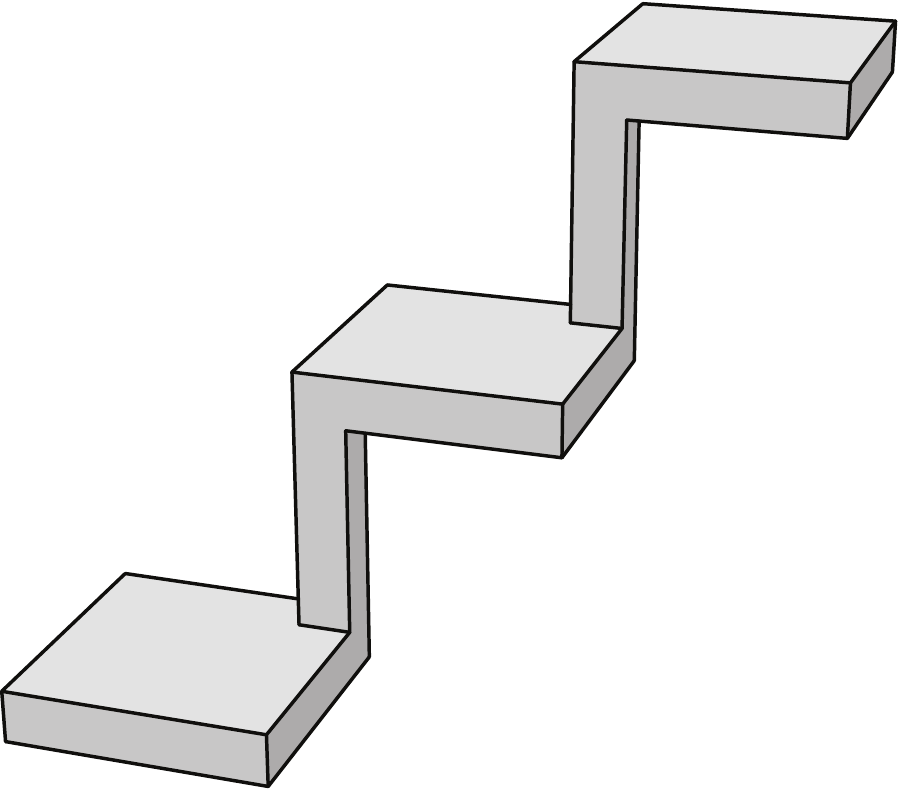}}
\caption{Orthogonal polyhedron that needs $\lfloor f/6\rfloor$ open face guards}
\label{fig4}
\end{figure}

Plugging $c=3$ in Theorem~\ref{t3:face} reveals that our lower bound is also tight.

\begin{theorem}
To guard an orthogonal polyhedron having $f$ faces, $\lfloor f/6\rfloor$ open face guards are always sufficient and occasionally necessary. \hfill $\square$
\end{theorem}

Moving on to \emph{closed} face guards in 4-oriented polyhedra, we propose the construction in Figure~\ref{fig5}. Each closed face sees the tip of at most one of the $k$ tetrahedral \emph{spikes}, hence $k$ guards are needed. Because there are $5k+2$ faces in total, a lower bound of $\lfloor f/5\rfloor$ closed face guards follows.

\begin{figure}[h]
\centering
\includegraphics[width=.75\linewidth]{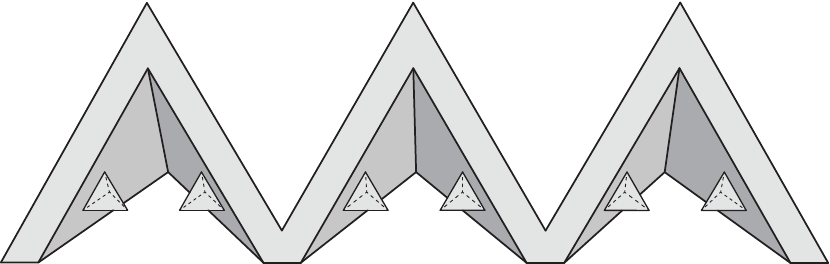}
\caption{4-oriented polyhedron that needs $\lfloor f/5\rfloor$ closed face guards}
\label{fig5}
\end{figure}

For \emph{open} face guards in 4-oriented polyhedra, we modify the previous example by carefully placing additional spikes on the other side of the construction, as Figure~\ref{fig6} illustrates. Once again, since each open face sees the tip of at most one of the $k$ spikes and there are $4k+2$ faces in total, a lower bound of $\lfloor f/4\rfloor$ open face guards follows.

\begin{figure}[h]
\centering
\includegraphics[width=.75\linewidth]{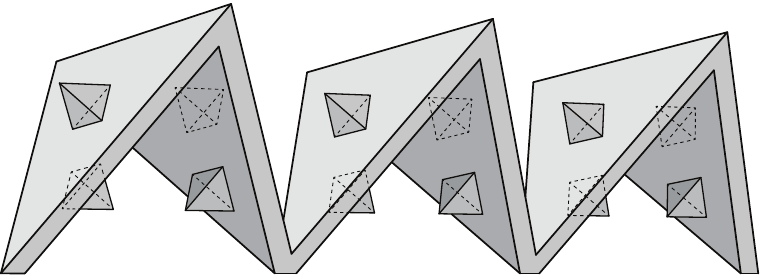}
\caption{4-oriented polyhedron that needs $\lfloor f/4\rfloor$ open face guards}
\label{fig6}
\end{figure}

This bound is also tight, as easily seen by plugging $c=4$ in Theorem~\ref{t3:face}.

\begin{theorem}
To guard a 4-oriented polyhedron having $f$ faces, $\lfloor f/4\rfloor$ open face guards are always sufficient and occasionally necessary. \hfill $\square$
\end{theorem}

\subsection{2-reflex orthostacks}

In our pursuit to lower the $\lfloor f/6\rfloor$ upper bound on closed face guards in orthogonal polyhedra, in order to match it with the $\lfloor f/7\rfloor$ lower bound, we study a special class of orthogonal polyhedra. Recall that the lower bound example in Figure~\ref{fig3} has no vertical reflex edges, but only reflex edges in two horizontal directions. These orthogonal polyhedra are called \emph{2-reflex}, and were first introduced by the author in~\cite{thesis}, and studied in conjunction with edge guards.

We further restrict our analysis to 2-reflex polyhedra that are also \emph{orthostacks}, as defined in~\cite{orthostacks}. An orthostack is an orthogonal polyhedron whose horizontal cross sections are simply connected. Therefore, a 2-reflex orthostack can be naturally viewed as a pile of cuboidal \emph{bricks} of various sizes, stacked on top of each other.

Our motivation for studying 2-reflex orthostacks is that they constitute the most obvious building block for 2-reflex polyhedra. In turn, 2-reflex polyhedra are a natural class of polyhedra of intermediate complexity between orthogonal prisms and orthogonal polyhedra. While 2-reflex orthostacks form a very basic class of polyhedra, they already pose some challenges, and we perceive them as a necessary and critical step toward solving the general face-guarding problem for orthogonal polyhedra.

\paragraph{Lower bounds.}

Observe that the polyhedron in Figure~\ref{fig4} is already a 2-reflex orthostack. Along with Theorem~\ref{t3:face}, this yields a tight bound of $\lfloor f/6\rfloor$ \emph{open} face guards in 2-reflex orthostacks.

On the other hand, the one in Figure~\ref{fig3}, despite being 2-reflex, is not an orthostack. However, the example in Figure~\ref{fig4} can be used once again to obtain a good lower bound on \emph{closed} face guards, as well.

\begin{figure}[h]
\centering
\includegraphics[width=.5\linewidth]{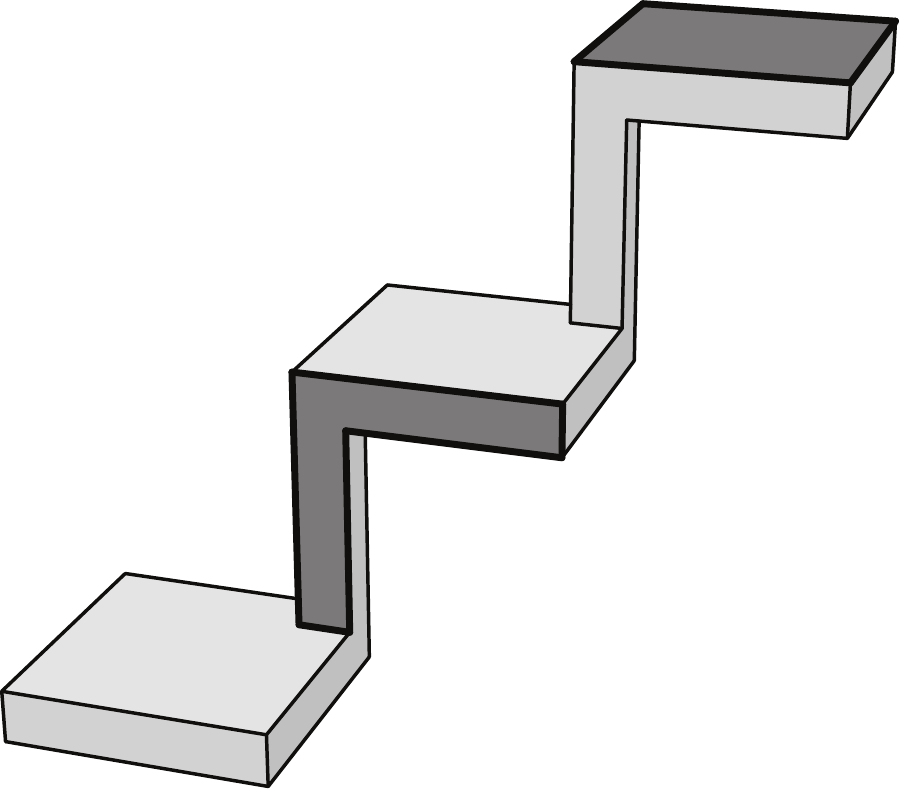}
\caption{2-reflex orthostack that needs $\lfloor (f+3)/9\rfloor$ closed face guards}
\label{figlowerortho}
\end{figure}

As Figure~\ref{figlowerortho} indicates, each triplet of consecutive bricks can be guarded by a single closed face guard. On the other hand, no closed face guard can see the centers of four bricks. If there are $k$ bricks in total, then there are $f=3k+3$ faces, and $g=\lceil k/3\rceil=\lfloor (k+2)/3\rfloor$ open face guards are needed. By substituting for $k$, we get $g=\lfloor (f+3)/9\rfloor$.

\paragraph{Upper bounds.}

We have already given a tight bound for open face guards, so let us consider closed face guards now. Note that two adjacent bricks of a 2-reflex orthostack share a horizontal rectangle, which we call the \emph{contact rectangle} between the two bricks. In general, each contact rectangle is coplanar with at least one horizontal face of the orthostack. If a contact rectangle is coplanar with exactly one face of the orthostack, such a contact rectangle is said to be \emph{canonical}.

Observe that any non-canonical contact rectangle between two bricks, having $k>1$ coplanar faces, can always be converted into $k$ canonical ones, by suitably adding new bricks between the initial two. More specifically, ``extruding'' the initial contact rectangle into a new brick, as shown in Figure~\ref{figstretch}, allows to separate the horizontal faces facing up from those facing down. Subsequently, if two coplanar up-facing (respectively, down-facing) faces are still present, one of them can be ``lifted'' (respectively, ``lowered'') via the addition of another brick.

\begin{figure}[h]
\centering
\includegraphics[width=.75\linewidth]{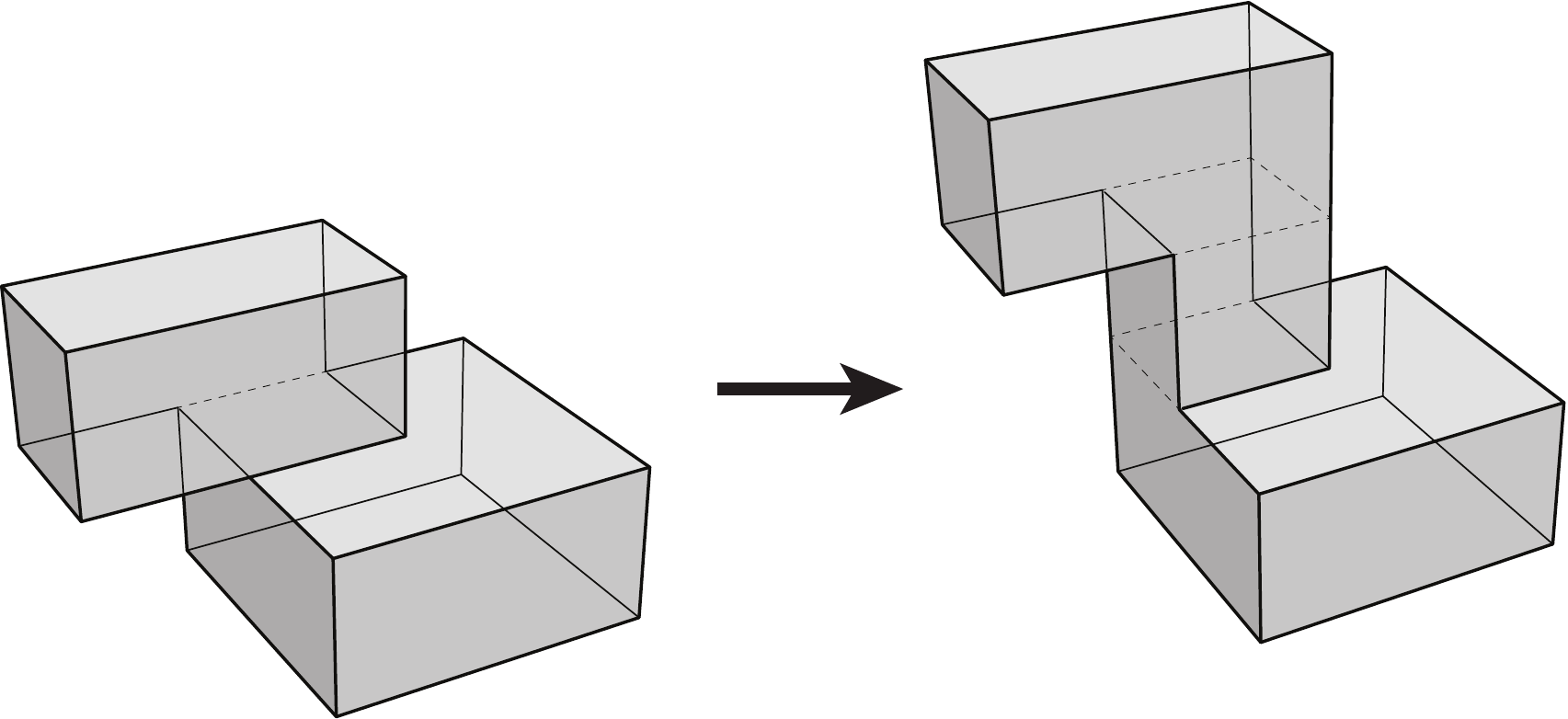}
\caption{Converting a generic contact rectangle into two canonical ones}
\label{figstretch}
\end{figure}

Note that this transformation does not change the number of faces of the polyhedron. Also, after the transformation, some visibilities between points may be lost, but none is gained. Therefore, if the resulting polyhedron is guardable by $g$ closed face guards, then the original polyhedron is guardable by the $g$ ``corresponding'' closed face guards.

Hence, in the following, we will restrict our analysis to 2-reflex orthostacks with canonical contact rectangles only. Depending to the relative positions of two bricks, the contact rectangle between them can be of one of four different \emph{types}, which are illustrated in Figure~\ref{fig2stack}. Each type is characterized by how many (reflex) edges are shared between the two adjacent bricks: a contact rectangle of type~$i$ has exactly $i$ edges of the orthostack on its perimeter, with $1\leqslant i\leqslant 4$. It is straightforward to see that there are no other possible configurations for a canonical contact rectangle.

\begin{figure}[h]
\centering
\subfigure[Type 1: $f=8$, $\Delta=4$]{\label{fig2stack:a}\includegraphics[scale=.45]{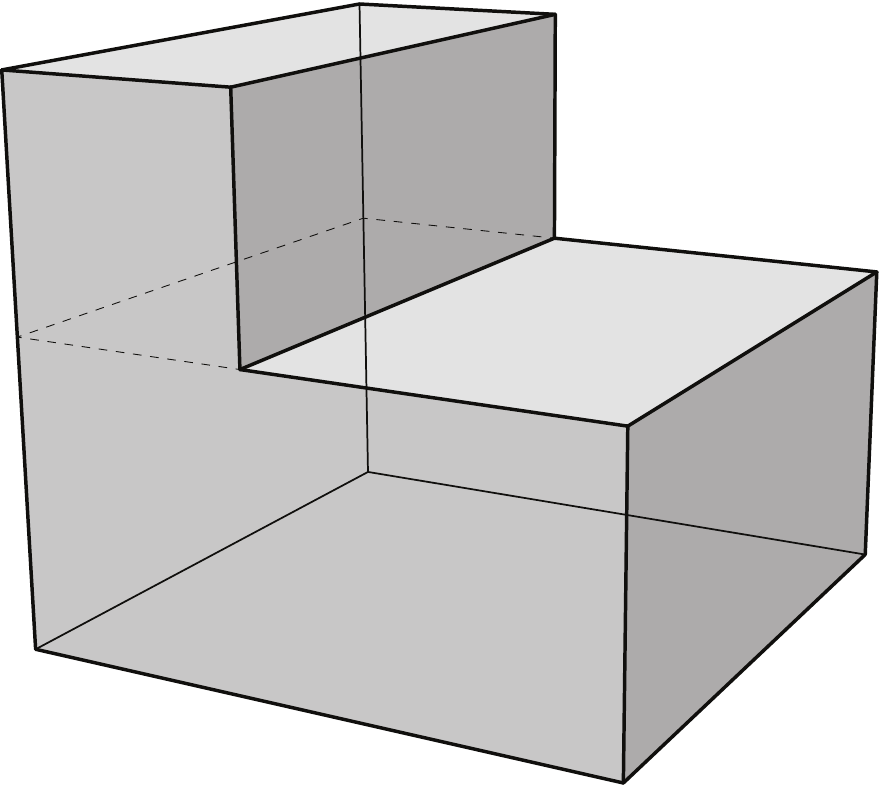}}\qquad\qquad\quad
\subfigure[Type 2: $f=9$, $\Delta=3$]{\label{fig2stack:c}\includegraphics[scale=.45]{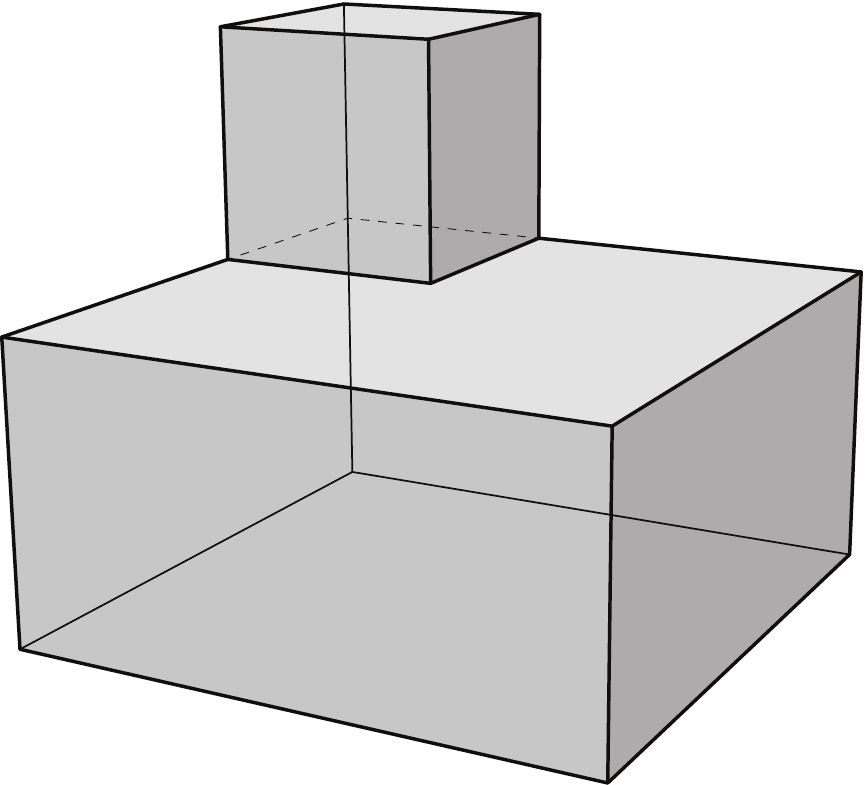}}\\
\subfigure[Type 3: $f=10$, $\Delta=2$]{\label{fig2stack:d}\includegraphics[scale=.45]{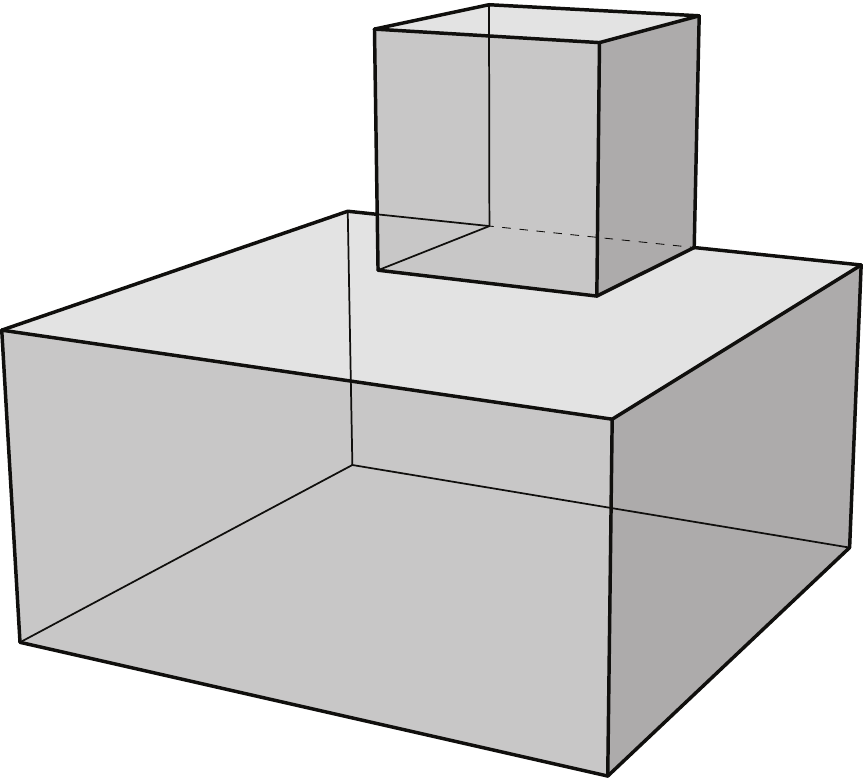}}\qquad\qquad\quad
\subfigure[Type 4: $f=11$, $\Delta=1$]{\label{fig2stack:e}\includegraphics[scale=.45]{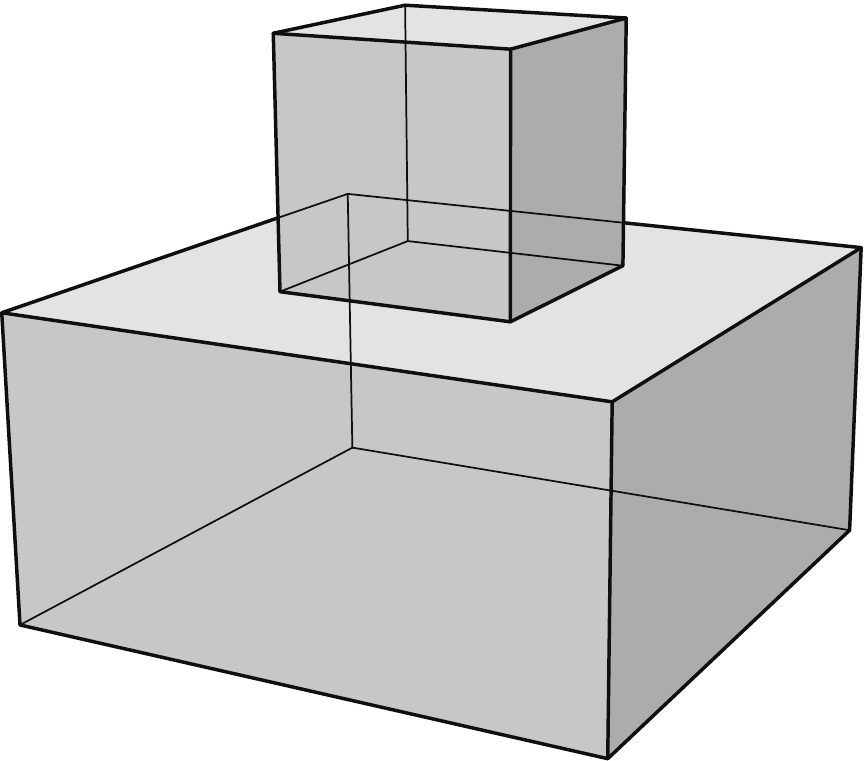}}
\caption{Types of canonical contact rectangles}
\label{fig2stack}
\end{figure}

When two bricks are attached on top of each other and a canonical contact rectangle is formed, some of their faces disappear or are merged together, and therefore the total number of faces decreases. The quantity by which it decreases is called the \emph{deficit} of the contact rectangle, and it only depends on its type. In Figure~\ref{fig2stack}, the number of faces of each orthostack consisting of two bricks is denoted by $f$, and the deficit of the corresponding contact rectangle is denoted by $\Delta$. Observe that $f+\Delta=12$, and that the deficit of a contact rectangle of type~$i$ is $\Delta=5-i$. It follows that the number of faces of a 2-reflex othostack with $k$ bricks depends only on the types of its $k-1$ contact rectangles, and this number is $6k$ minus the sum of the deficits of the contact rectangles. This is also equal to $k+5$ plus the sum of the types of the contact rectangles.

Let two adjacent bricks be given, sharing a canonical contact rectangle of type~$i$, with $1\leqslant i\leqslant 4$. Note that the vertical projection of one of the two bricks is strictly contained in the vertical projection of the other brick. If the vertical projection of the lower (respectively, upper) brick is contained in the vertical projection of the upper (respectively, lower) brick, we denote the configuration by the symbol $\ls{i}$ (respectively, $\rs{i}$). Similarly, we indicate a stack of arbitrarily many bricks by a sequence of labeled $\mathlarger\bigsqcup$ and $\mathlarger\bigsqcap$ symbols, and we call this sequence the \emph{signature} of the orthostack. For instance, the signature of the orthostack in Figure~\ref{figlowerortho} is $$\rs{2}\ls{2}\rs{2}\ls{2}.$$

The symbol $\lrs{i}$ is shorthand for ``$\ls{i}$ or $\rs{i}$''. Similarly, $\lrs{i}\lrs{j}$ stands for $$\mbox{``$\ls{i}\ls{j}$ or $\ls{i}\rs{j}$ or $\rs{i}\ls{j}$ or $\rs{i}\rs{j}$'',}$$ and so on.

As it turns out, to obtain our upper bound on closed face guards, we can almost entirely abstract from the actual shapes of 2-reflex orthostacks, and just reason about their signatures.

\begin{proposition}
Any 2-reflex orthostack whose signature is of the form $\rs{i}\rs{j}$, with $j\neq 4$, is guardable by one vertical closed face guard.
\end{proposition}
\begin{proof}
If the top contact rectangle is not of type~4, there is a vertical face that is shared by the two top bricks, as Figure~\ref{fig7:b} exemplifies. This face also touches the bottom brick, and so it guards the entire orthostack.
\end{proof}

\begin{proposition}
Any 2-reflex orthostack with signature of the form $\rs{i}\ls{j}$ is guardable by one vertical closed face guard.
\end{proposition}
\begin{proof}
It is sufficient to choose any vertical face of the middle brick. As Figure~\ref{fig7:c} suggests, this face sees also the top and bottom bricks.
\end{proof}

\begin{figure}[h]
\centering
\subfigure[$\rs{2}\rs{3}$]{\label{fig7:b}\includegraphics[scale=.4]{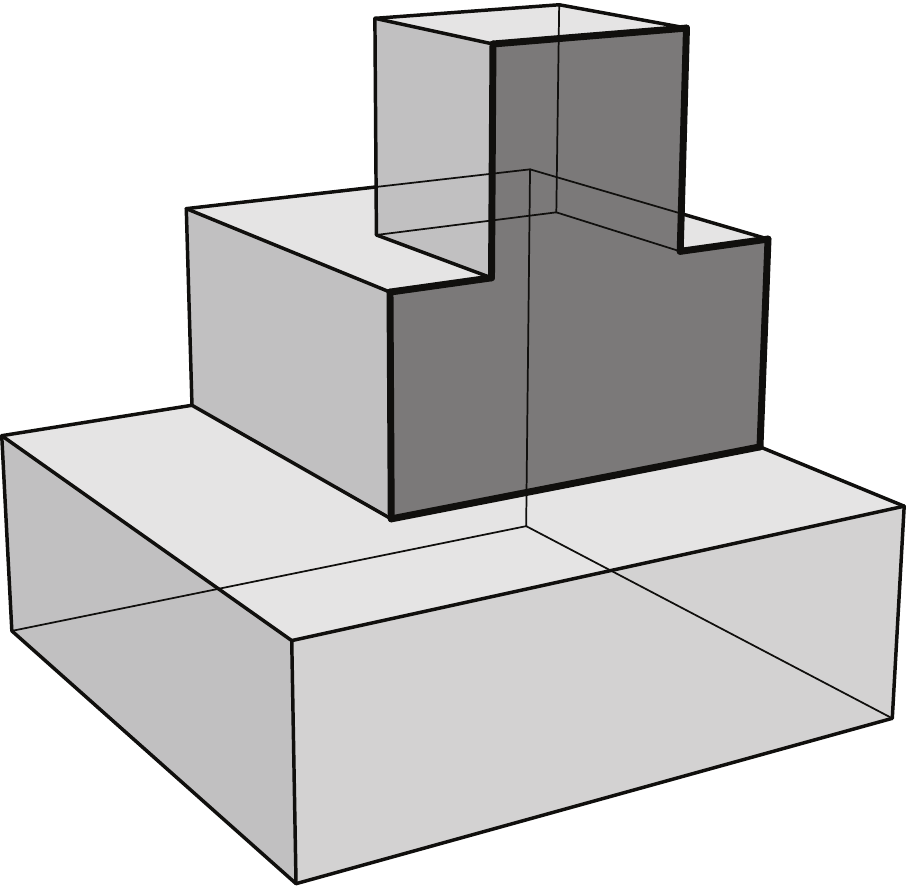}}\qquad
\subfigure[$\rs{4}\ls{4}$]{\label{fig7:c}\includegraphics[scale=.4]{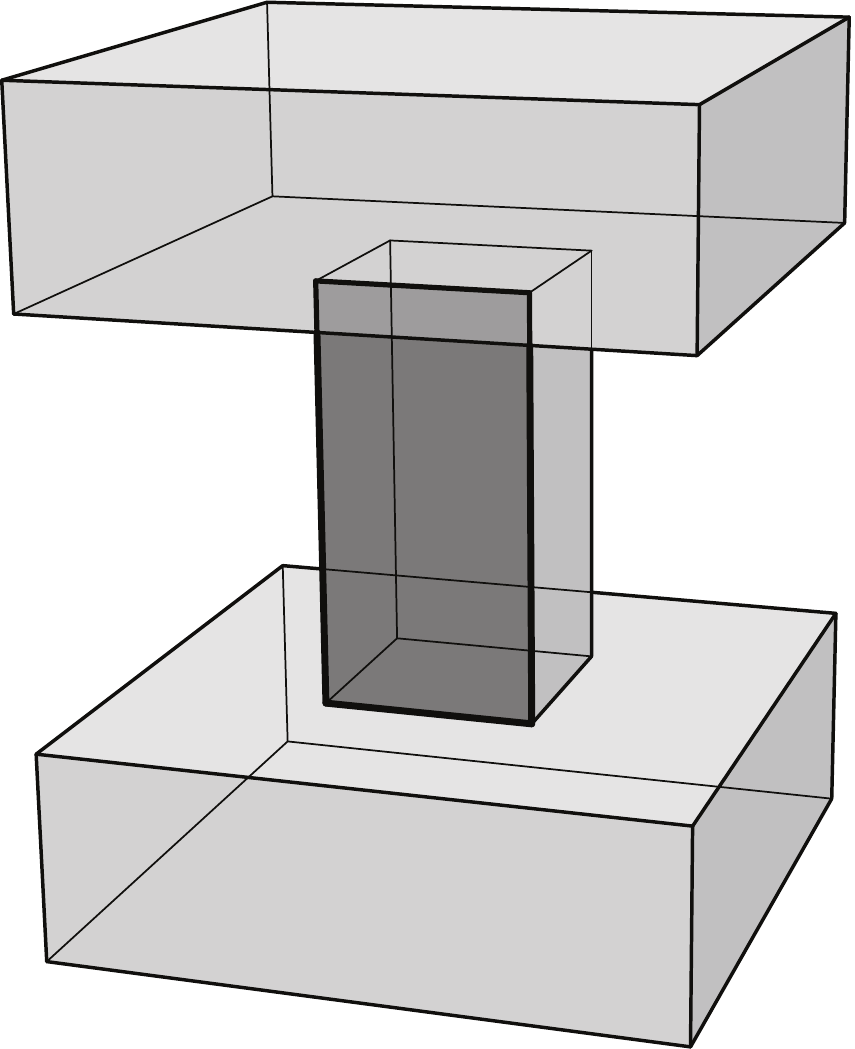}}\qquad
\subfigure[$\ls{1}\rs{2}$]{\label{fig7:d}\includegraphics[scale=.4]{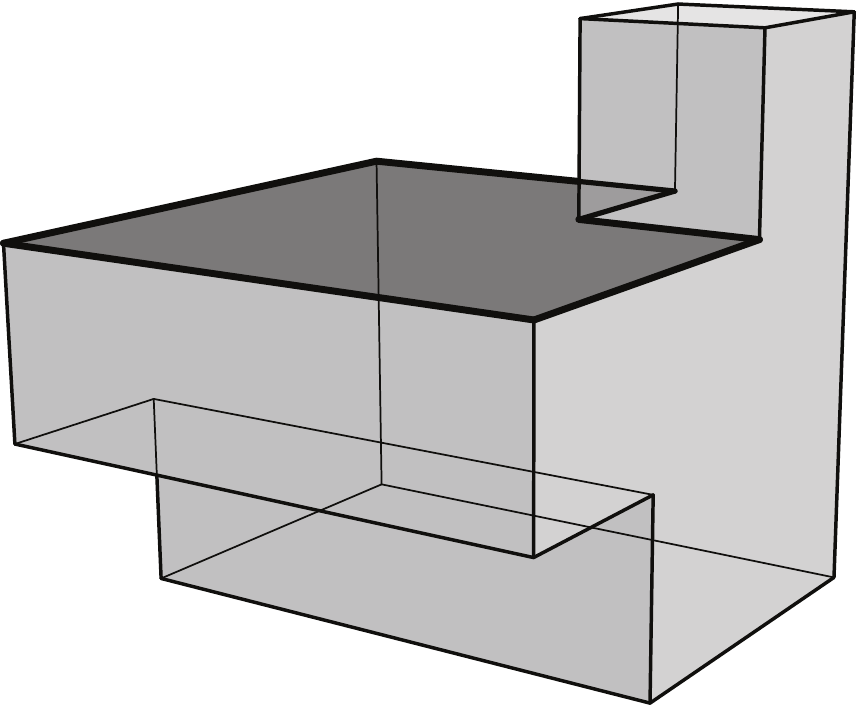}}
\caption{Guarding some 3-brick orthostacks with a single guard}
\label{fig7}
\end{figure}

\begin{proposition}
Any 2-reflex orthostack with signature of the form $\ls{i}\rs{j}$ is guardable by one horizontal closed face guard that is not the topmost nor the bottommost face of the orthostack.
\end{proposition}
\begin{proof}
Let us consider the two horizontal faces that border the middle brick: we show that one of them guards the whole polyhedron. If the vertical projection of the top brick is entirely contained in the vertical projection of the bottom brick, as Figure~\ref{fig7:d} shows, we pick the top face of the middle brick. Otherwise, the vertical projection of the top brick intersects the bottom face of the middle brick. Hence we may pick this face, as it guards all three bricks.
\end{proof}

From the three previous propositions, we straightforwardly obtain the following.

\begin{lemma}\label{l:stack1}
Any 2-reflex orthostack made of three bricks, whose signature is neither of the form $\rs{i}\rs{4}$ nor $\ls{4}\ls{i}$, is guardable by one closed face guard that is not the topmost nor the bottommost horizontal face of the orthostack.\hfill\qed
\end{lemma}

\begin{lemma}\label{l:stack2}
Any 2-reflex orthostack with signature $\lrs{1}\lrs{1}\lrs{1}$ is guardable by one vertical closed face guard.
\end{lemma}
\begin{proof}
Figure~\ref{figeeee} shows the construction of such an orthostack, brick by brick. When only the bottom brick is present, the guard can be chosen among its four vertical faces. Each time a new brick is added on top of the old ones, exactly three vertical faces are extended to form three sides of the new brick. As a consequence, after three bricks have been added on top of the first one, there is still at least one vertical face that stretches from the very bottom to the very top of the construction. This face guards the entire orthostack.
\end{proof}

\begin{figure}[h]
\centering
\subfigure[]{\label{figeeee:a}\includegraphics[scale=.35]{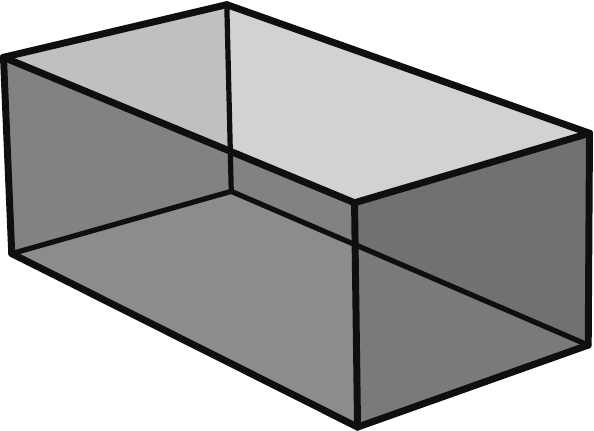}}\qquad
\subfigure[]{\label{figeeee:b}\includegraphics[scale=.35]{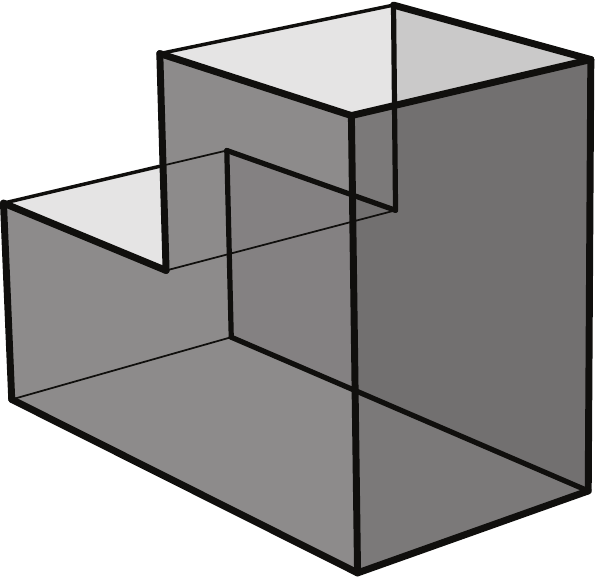}}\qquad
\subfigure[]{\label{figeeee:c}\includegraphics[scale=.35]{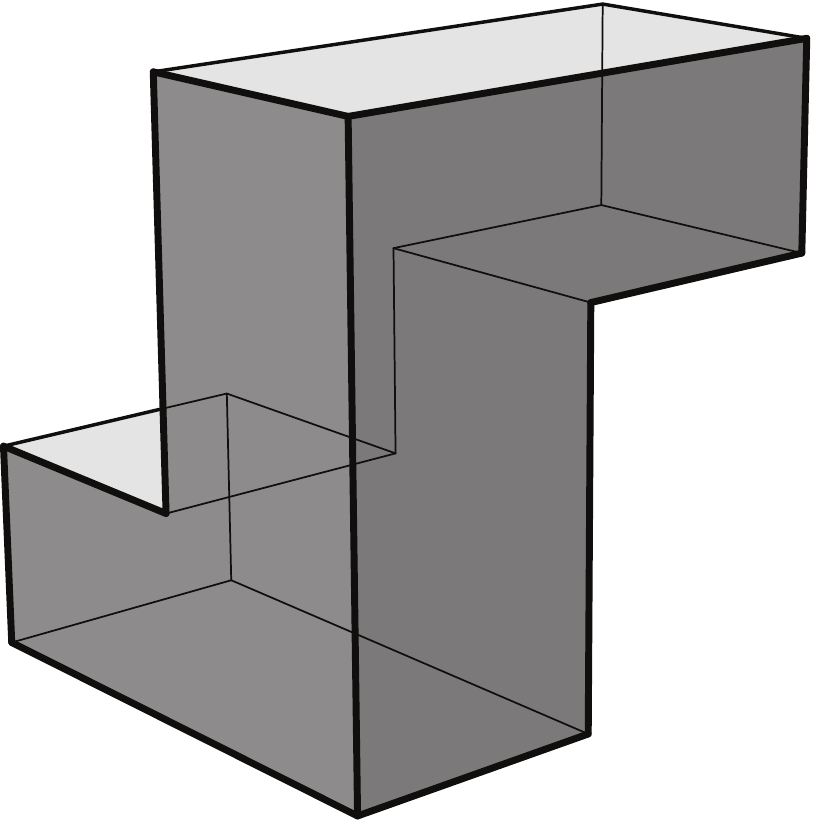}}\quad
\subfigure[]{\label{figeeee:d}\includegraphics[scale=.35]{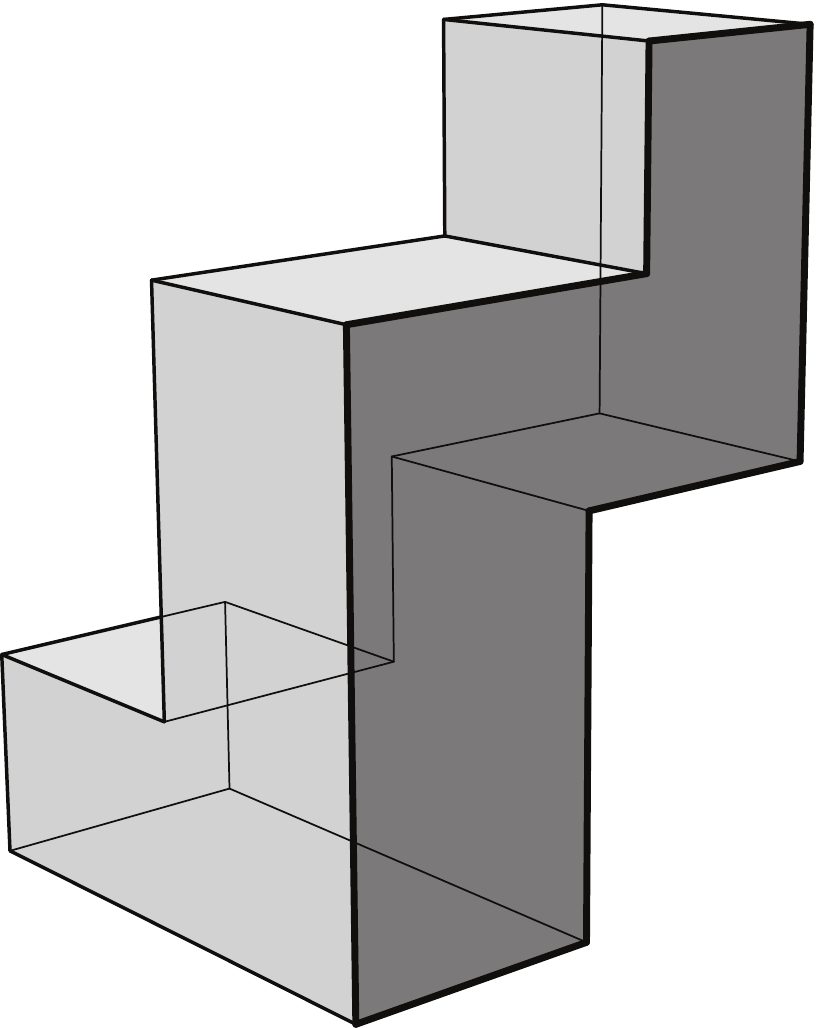}}
\caption{Guarding a $\rs{1}\ls{1}\rs{1}$ with a single guard}
\label{figeeee}
\end{figure}

\begin{theorem}
Any 2-reflex orthostack with $f$ faces is guardable by
$$\left\lfloor \frac{f+1} 7\right\rfloor$$
closed face guards.
\end{theorem}
\begin{proof}
We will prove a slightly stronger statement: the $\lfloor (f+1)/7\rfloor$ face guards can be chosen in such a way that none of them lies on the topmost horizontal face of the orthostack. We prove this claim by induction on the number of bricks. Given a 2-reflex orthostack $\mathcal P$ with $k\geqslant 0$ bricks, suppose that the claim holds for all 2-reflex orthostacks with fewer bricks, and let us prove that it holds for $\mathcal P$, as well.

The cases $k=0$ and $k=1$ are trivial: if $k=0$, then $f=0$, and zero guards are sufficient; if $k=1$, then $f=6$, and $\mathcal P$ is guarded by any vertical face. For the case $k=2$, the possible configurations are represented in Figure~\ref{fig2stack}, which shows that $8\leqslant f\leqslant 11$, and $\mathcal P$ is easily guardable by a single vertical face. If $k=3$, Lemma~\ref{l:stack1} guarantees that one vertical face guard is sufficient, unless the signature of $\mathcal P$ is of the form $\rs{i}\rs{4}$ or $\ls{4}\ls{i}$. But in this case $f\geqslant 13$, so we are allowed to place two face guards, and it is easy to find two vertical faces that guard all the three bricks.

Now suppose that $k\geqslant 4$, and let $1\leqslant j\leqslant k$. Let $\mathcal P'$ be the orthostack formed by the $j$ topmost bricks of $\mathcal P$, having $f'$ faces, and let $\mathcal P''$ be the rest of the orthostack, made of $k-j$ bricks and having $f''$ faces. The total number of faces of $\mathcal P$ is $f=f'+f''-\Delta$, where $\Delta$ is the deficit of the contact rectangle between $\mathcal P'$ and $\mathcal P''$ (if $j=k$, we may take $\Delta=0$). By inductive hypothesis, we can guard $\mathcal P''$ with at most $\lfloor (f''+1)/7\rfloor$ closed face guards, none of which lies on the topmost horizontal face. Suppose that $\mathcal P'$ can be guarded with at most $\lfloor (f'-\Delta)/7\rfloor$ closed face guards, none of which lies on the topmost or the bottommost horizontal face. In total, we would have at most
$$\left\lfloor \frac{f''+1} 7\right\rfloor + \left\lfloor \frac{f'-\Delta} 7\right\rfloor \leqslant \left\lfloor \frac{f''+1+f'-\Delta} 7\right\rfloor = \left\lfloor \frac{f+1} 7\right\rfloor$$
closed face guards. Moreover, because none of these face guards would be coplanar with the contact rectangle between $\mathcal P'$ and $\mathcal P''$, they would be naturally mapped into face guards of $\mathcal P$. Indeed, the horizontal guards maintain the same shape and size after the merge, while some vertical guards may be merged with other faces, and thus enlarged, which is not an issue because this only makes them guard a bigger area. Together, these faces would guard all of $\mathcal P$, none of them would lie on its topmost horizontal face, and therefore our main claim on $\mathcal P$ would be proven.

Let us show that, in every case, it is always possible to choose $j$ in such a way that the desired conditions on $\mathcal P'$ are met, allowing our previous reasoning to go through. In most cases, choosing $j=2$ is enough, as detailed next. Let $a$ and $b$ be the types of the two topmost contact rectangles, and let $\Delta'$ be the deficit of the topmost contact rectangle. So, if $\mathcal P'$ consists of two bricks, $f'=12-\Delta'$. Now, if $\lfloor (f'-\Delta)/7\rfloor\geqslant1$, we are allowed to place at least one guard in $\mathcal P'$, and it is easy to see that one vertical closed face guard is always sufficient (cf.~Figure~\ref{fig2stack}). Hence we want $f'-\Delta\geqslant 7$ to hold, which is equivalent to $\Delta+\Delta'\leqslant 5$, that is, $a+b\geqslant 5$.

The only cases left are those in which $a+b\leqslant 4$. Namely, these are the cases in which the signature of the three topmost bricks is one of the following (or one of their reverses): $\lrs{1}\lrs{1}$, $\lrs{1}\lrs{2}$, $\lrs{1}\lrs{3}$, $\lrs{2}\lrs{2}$. In the last three cases, we choose $j=3$. Indeed, in these cases $f'$ is either 11 or 12, and $\lfloor (f'-\Delta)/7\rfloor=1$. By Lemma~\ref{l:stack1}, we can guard $\mathcal P'$ with one closed face guard that has the desired properties.

Finally, let us assume that the signature of the three topmost bricks of $\mathcal P$ is $\lrs{1}\lrs{1}$. Choosing $j=3$ works as above, unless the third contact rectangle of $\mathcal P$ is again of type~1 (indeed, if the type is at least 2, we have $f'=10$, $\Delta\leqslant 3$, and therefore $\lfloor (f'-\Delta)/7\rfloor\geqslant 1$). In this last case, we choose $j=4$ (recall that $k\geqslant 4$), so that the signature of $\mathcal P'$ is $\lrs{1}\lrs{1}\lrs{1}$. We have $f'=12$, $0\leqslant \Delta\leqslant 4$ ($\Delta=0$ holds if $k=4$), and $\lfloor (f'-\Delta)/7\rfloor=1$. By Lemma~\ref{l:stack2}, $\mathcal P'$ can be guarded by a single vertical guard, and our theorem follows.
\end{proof}

\section{Minimizing face guards}\label{s4}

\subsection{Hardness of approximation}

In~\cite{faceguards}, Souvaine et al.\ ask for the complexity of minimizing face guards in a given polyhedron. We show that this problem is at least as hard as \SET, and we infer that approximating the minimum number of (closed or open) face guards within a factor of $\Omega(\log f)$ is \NP-hard. This remains true even if we restrict the problem to the class of simply connected orthogonal polyhedra.

We also show that the same hardness of approximation result holds for non-triangulated terrains. Recall that, in~\cite{terrain1}, Iwamoto et al.\ proved that minimizing closed face guards in triangulated terrains is \NP-hard. Thus, we improve on their result in the case of non-triangulated terrains, while also extending it to open face guards.

\paragraph{Orthogonal polyhedra.}

We give a linear approximation-preserving reduction from \SET, in the sense of~\cite[Definition~8.4]{ausiello}.

\begin{theorem}\label{t3:hard}
\SET is L-reducible to the problem of minimizing (closed or open) face guards in a simply connected orthogonal polyhedron.
\end{theorem}
\begin{proof}
Let an instance of \SET be given, i.e., a \emph{universe} $U=\{1,\cdots,n\}$, and a collection $S\subseteq \mathcal P(U)$ of $m\geqslant 1$ subsets of $U$. We will construct a simply connected orthogonal polyhedron with $f\in O(mn)$ faces that can be guarded by $k$ (closed or open) faces if and only if $U$ is the union of $k-1$ elements of $S$.

Figure~\ref{figset2} shows our construction for $U=\{1,2,3,4\}$ and $S=\{\{2,4\},\{1,3\},\{2\}\}$. Figure~\ref{figset1} illustrates the side view of a generic case in which $m=4$.

\begin{figure}[h]
\centering
\includegraphics[width=.75\linewidth]{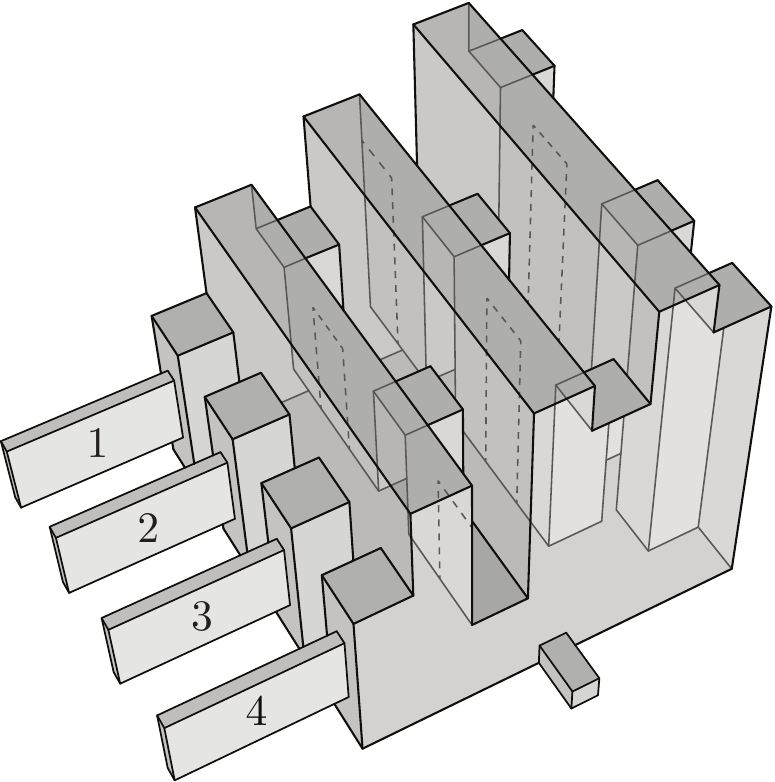}
\caption{\SET reduction for orthogonal polyhedra, 3D view}
\label{figset2}
\end{figure}

\begin{figure}[h]
\centering
\includegraphics[width=.75\linewidth]{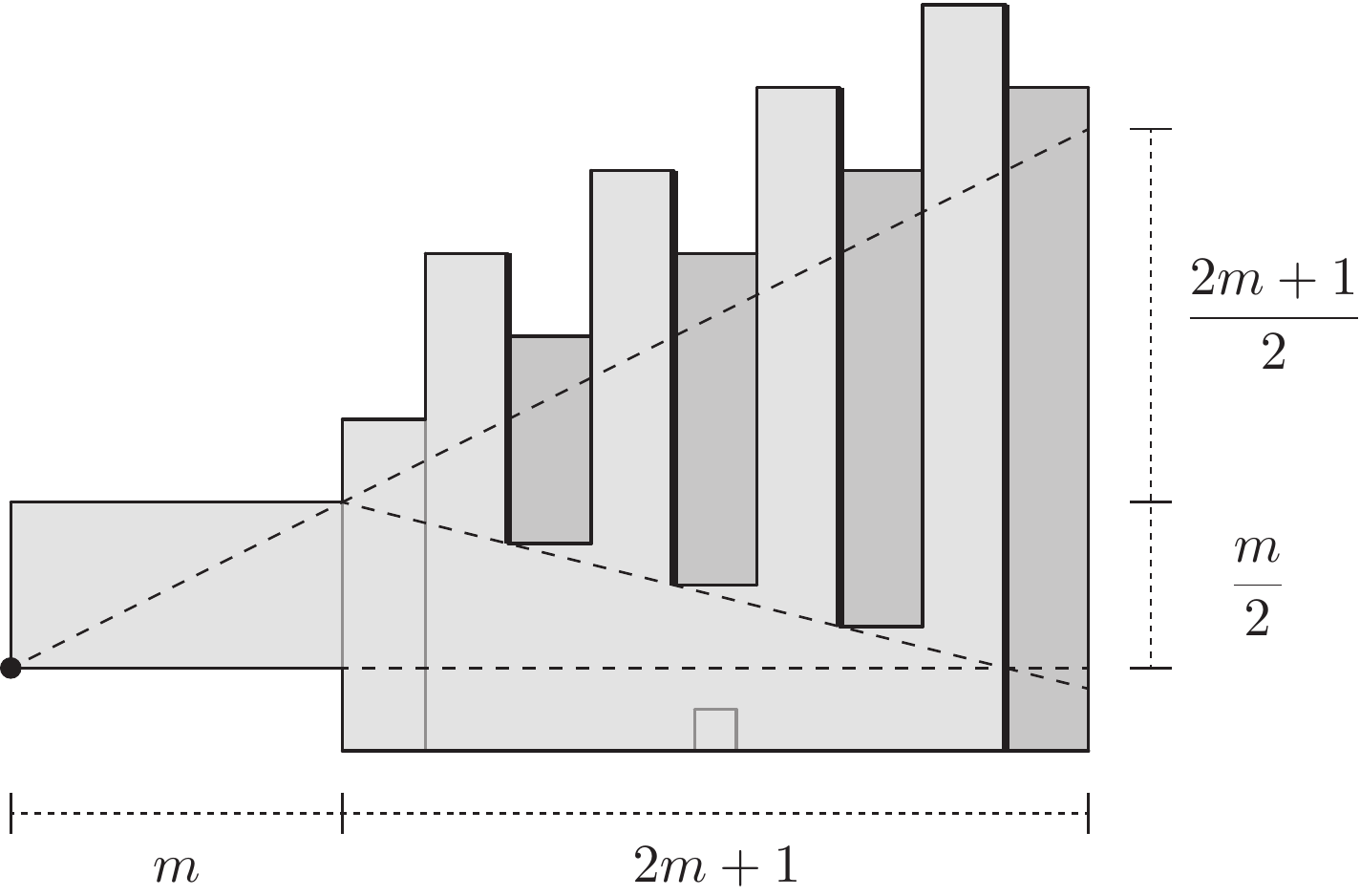}
\caption{\SET reduction for orthogonal polyhedra, side view}
\label{figset1}
\end{figure}

Each of the thin cuboids on the far left is called a \emph{fissure}, and represents an element of $U$. Facing the fissures there is a row of $m$ \emph{mountains} of increasing height, separated by \emph{valleys} of increasing depth. The $m$ vertical walls that are facing the fissures (drawn as thick lines in Figure~\ref{figset1}) are called \emph{set faces}, and each of them represents an element of $S$.

For each $S_i\in S$, we dig a narrow rectangular \emph{dent} in the $i$-th set face in front of the $j$-th fissure, if and only if $j\notin S_i$. Each dent reaches the bottom of its set face, and almost reaches the top, so that it does not separate the set face into two distinct faces. Moreover, every dent (except those in the rightmost set face) is so deep that it connects two neighboring valleys. In Figure~\ref{figset1}, dents are depicted as darker regions; in Figure~\ref{figset2}, the dashed lines mark the areas where dents are \emph{not} placed.

We want to fix the width of the fissures in such a way that only a restricted number of faces can see their bottom. Specifically, consider $n$ \emph{distinguished points}, located in the middle of the lower-left edges of the fissures (indicated by the thick dot in Figure~\ref{figset1}). The $j$-th distinguished point definitely sees some portions of the $i$-th set face, provided that $j\in S_i$. If this is the case, and $i<m$, it also sees portions of two other faces (one horizontal, one vertical) surrounding the same valley. Moreover, if $j\notin S_m$, the $j$-th distinguished point also sees the bottom of a dent in the rightmost set face. We want no face to be able to see any distinguished point, except the faces listed above (plus of course the faces belonging to fissures or surrounding their openings). To this end, assuming that the dents have unit width, we set the width of the fissures to be slightly less than $1/4$. Indeed, referring to Figure~\ref{figset1}, the width of the visible region of a distinguished point, as it reaches the far right of the construction, is strictly less than $$\frac{(m)+(2m+1)}{m}\cdot\frac 14\ =\ \left(3+\frac 1m\right)\cdot\frac 14\ \leqslant\ 4\cdot\frac 14\ =\ 1,$$ because $m\geqslant 1$.

Finally, a small \emph{niche} is added in the lower part of the construction. Its purpose is to enforce the selection of a ``dedicated'' face guard, as no face can see both a distinguished point and the bottom of the niche.

Let a guarding set for our polyhedron be given, consisting of $k$ face guards. We will show how to compute in polynomial time a solution of size at most $k-1$ for the given \SET instance, provided that it is solvable at all.

We first discard every face guard that is not guarding any distinguished point. Because at least one such face must guard the niche, we are left with at most $k-1$ guards. Then, if any of the remaining face guards borders the $i$-th valley, with $i<m$, we replace it with the $i$-th set face. Indeed, it is easy to observe that such a set face can see the same distinguished points, plus possibly some more. By construction, all the remaining guards can see exactly one distinguished point (they are either faces belonging to some fissure, or surrounding its opening, or rightmost faces of the rightmost dents). We replace each of these face guards with any set face that guards the same distinguished point (which exists, otherwise the \SET instance would be unsolvable). As a result, we have at most $k-1$ set faces guarding all the distinguished points. These immediately determine a solution of equal size to the given \SET instance.

Conversely, if the \SET instance has a solution of size $k$, it is easy to see that our polyhedron has a guarding set of $k+1$ guards: all the set faces corresponding to the \SET's solution, plus the bottom face.
\end{proof}

\begin{cor}\label{cor:hard}
Given a simply connected orthogonal polyhedron with $f$ faces, it is \NP-hard to approximate the minimum number of (closed or open) face guards within a factor of $\Omega(\log f)$.
\end{cor}
\begin{proof}
The polyhedra constructed by the L-reduction of Theorem~\ref{t3:hard} have $f\in O(mn)$ faces. It was proved in~\cite{set} that \SET is \NP-hard to approximate within a ratio of $\Omega(\log n)$ and, by inspecting the reduction employed, it is apparent that all the hard \SET instances generated are such that $m\in O(n^c)$, for some constant $c\geqslant 1$. As a consequence, we may assume that $\Omega(\log n)=\Omega(\log n^{c+1})\subseteq \Omega(\log(mn))\subseteq \Omega(\log f)$. Since the minimum is \NP-hard to approximate within some factor belonging to $\Omega(\log n)$, and the same factor also belongs to $\Omega(\log f)$, our claim follows.
\end{proof}

\begin{observation}
An analogous of Theorem~\ref{t3:hard} and of Corollary~\ref{cor:hard} also holds, with the same proof, for the related problem of guarding the \emph{boundary} of a polyhedron by face guards, as opposed to the whole interior.
\end{observation}

\paragraph{Non-triangulated terrains.}

We show that the above reduction can be adapted to work with non-triangulated terrains. A \emph{terrain} is a piecewise-linear surface embedded in $\R^3$, such that any vertical line intersects the surface in exactly one point. Therefore, a terrain is an unbounded 2-manifold that partitions $\R^3$ in an \emph{upper region} and a \emph{lower region}, each of which homeomorphic to a half-space. Faces, vertices and edges of terrains are defined in the same way as for polyhedra (cf.~Section~\ref{s2}). We stipulate that a terrain has exactly one unbounded face, and therefore no unbounded edges.

Visibility is defined in the upper region only: two points belonging to the upper region are visible if the line segment connecting them does not intersect the lower region. Therefore, the Art Gallery Problem for face guards in terrains asks for a set of faces that collectively see the whole upper region of a given terrain. The problem of guarding terrains is connected to the problem of guarding polyhedra, in that terrains may be intuitively viewed as a class of ``upward-unbounded polyhedra''.

A terrain whose bounded faces are triangles is called a \emph{triangulated} terrain. These special terrains are studied in~\cite{terrain2,terrain1,terrain3}, where it is shown that computing the minimum number of closed face guards in a given triangulated terrain is \NP-hard. Here we strengthen this result by showing that such a minimum is even \NP-hard to approximate within a logarithmic factor, for both open and closed face guards, provided that terrains are not necessarily triangulated.

\begin{theorem}\label{t4:hard}
Given a (not necessarily triangulated) terrain with $f$ faces, it is \NP-hard to approximate the minimum number of (closed or open) face guards within a factor of $\Omega(\log f)$.
\end{theorem}
\begin{proof}
We show that \SET is L-reducible to the problem of minimizing (closed or open) face guards in a non-triangulated terrain, by suitably modifying the construction given in Theorem~\ref{t3:hard}. Then, our claim follows as in Corollary~\ref{cor:hard}.

Our new construction is sketched in Figures~\ref{figter1} and~\ref{figter2}, again for $U=\{1,2,3,4\}$ and $S=\{\{2,4\},\{1,3\},\{2\}\}$. The faces that look vertical in Figure~\ref{figter1} are actually very steep slopes, as the side view of Figure~\ref{figter2} suggests.

\begin{figure}[h]
\centering
\includegraphics[width=.75\linewidth]{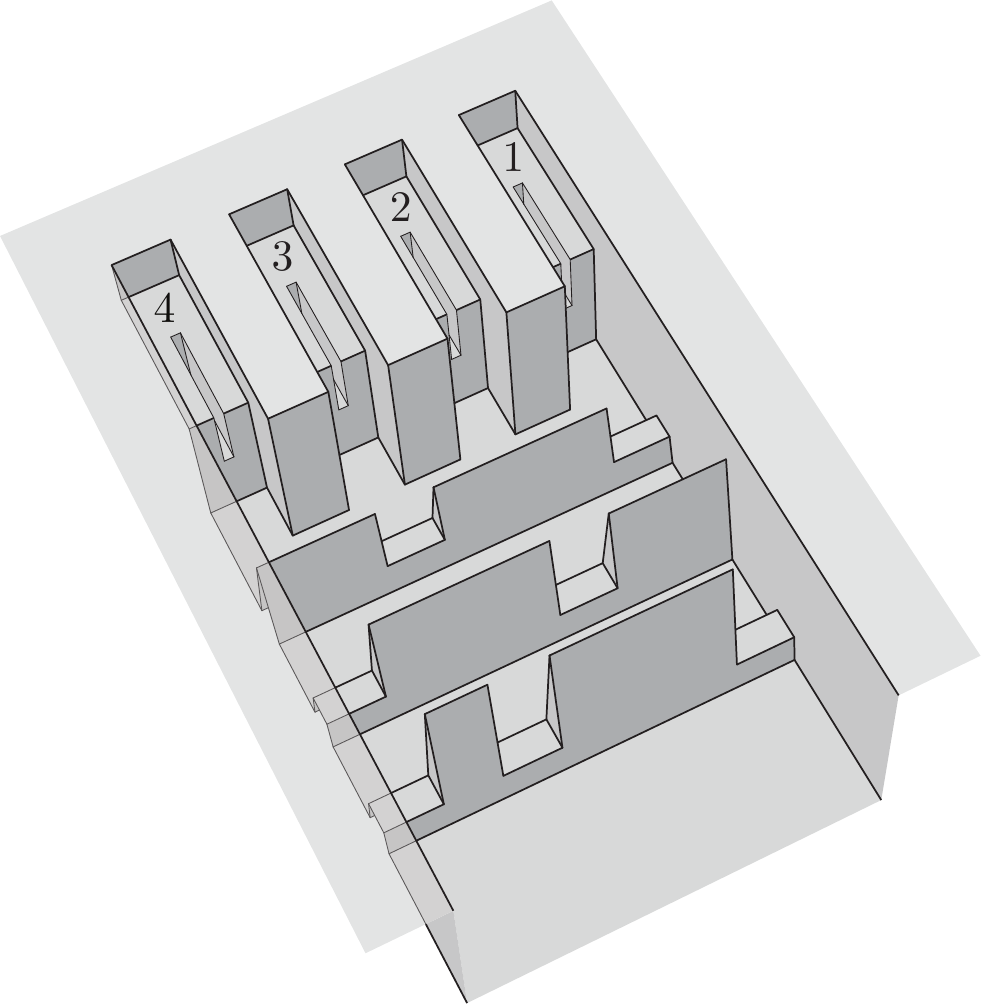}
\caption{\SET reduction for non-triangulated terrains, 3D view}
\label{figter1}
\end{figure}

We have $n$ very thin \emph{fissures} and $m$ \emph{mountains} of increasing height. The proportions are chosen in such a way that no face in the terrain can see inside two distinct fissures all the way to the far corner (i.e., the \emph{distinguished point} in Figure~\ref{figter2}), except the \emph{set faces} on the mountains. In particular, the construction is so ``long'' that the wall opposite to the fissures can see no fissure, as its visual is obstructed by the mountains (refer to the dashed line in Figure~\ref{figter2}). Moreover, each mountain, due to its \emph{dents}, can see exactly the fissures that correspond to the subset of $U$ that the mountain itself represents. Observe that, once again, $f\in O(mn)$.

\begin{figure}[h]
\centering
\includegraphics[width=.85\linewidth]{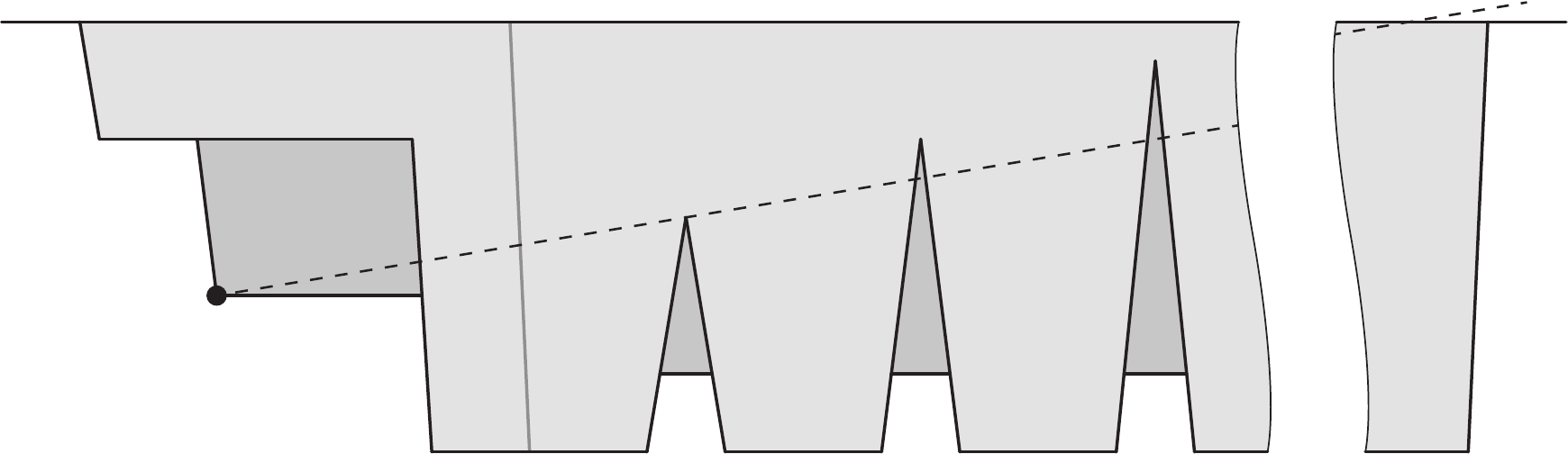}
\caption{\SET reduction for non-triangulated terrains, side view}
\label{figter2}
\end{figure}

Now, regardless of whether face guards are closed or open, all the remaining parts of the terrain can be guarded by a number of face guards that is bounded by a small constant $c$. Indeed, if the mountains are thin enough, all the dents can be collectively guarded by the two large side walls of the terrain (i.e., the light-shaded polygons in Figure~\ref{figter2}). Because we may assume that the ``hard'' \SET instances have arbitrarily large optimal solutions, $c$ becomes negligible in the computation of the approximation ratio, and our theorem follows.
\end{proof}

\subsection{Computing visible regions}

The next natural question is whether the minimum number of face guards can be computed in \NP, and possibly approximated within a factor of $\Theta(\log f)$ in polynomial time. Usually, when finitely many possible guard locations are allowed (such as with vertex guards and edge guards), this is established by showing that the visible region of any guard can be computed efficiently, as well as the intersection of two visible regions, etc. As a result, the environment is partitioned into polynomially many regions such that, for every region $R$ and every guard $g$, either $R\subseteq \mathcal V(g)$ or $R\cap \mathcal V(g)=\varnothing$. This immediately leads to a reduction to \SET, which yields an approximation algorithm with logarithmic ratio, via a well-known greedy heuristic~\cite{ghosh}.

With face guards (and also with edge guards in polyhedra) the situation is complicated by the fact that the visible region of a guard may not be a polyhedron, but in general its boundary is a piecewise quadric surface.

For example, consider the orthogonal polyhedron in Figure~\ref{figvis}. It is easy to see that the visible region of the bottom face (and also the visible region of edge $a$) is the whole polyhedron, except for a small region bordered by the thick dashed lines.

\begin{figure}[h]
\centering
\includegraphics[width=.65\linewidth]{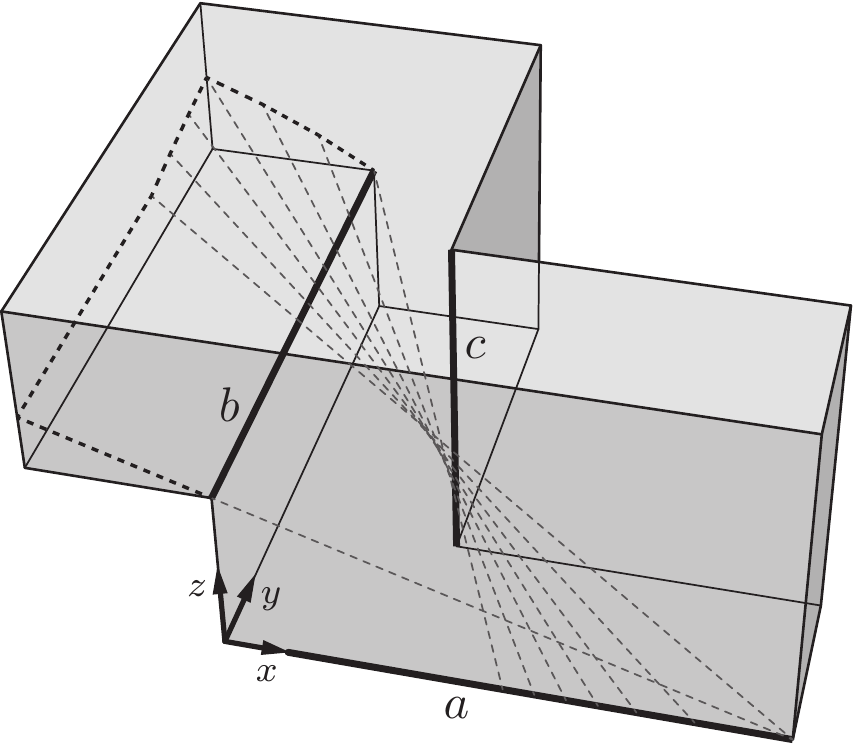}
\caption{The visible region of the bottom face is bounded by a hyperboloid of one sheet.}
\label{figvis}
\end{figure}

The surface separating the visible and invisible regions consists of a right trapezoid plus a bundle of mutually skew segments whose extensions pass through the edges $a$, $b$, and $c$. These edges lie on three lines having equations $$y^2+z^2=0,$$ $$x^2+(z-1)^2=0,$$ $$(x-1)^2+(y-1)^2=0,$$ respectively. A straightforward computation reveals that the  bundle of lines passing through these three lines has equation $$xy-xz+yz-y=0,$$ which defines a hyperboloid of one sheet.

In general, the boundary of the visible area of a face (or an edge) is determined by lines passing through pairs or triplets of edges of the polyhedron. If three edges are all parallel to a common plane, the surface they determine is a hyperbolic paraboloid (degenerating into a plane if two of the edges are parallel to each other), otherwise they determine a hyperboloid of one sheet, as in the above example.

There exists an extensive literature of purely algebraic methods to compute intersections of quadric surfaces (see for instance~\cite{quadric}), but the parameterizations involved may yield coefficients containing radicals. At this stage in our understanding, it is not clear whether any of these methods can be effectively applied to reduce the minimization problem of face-guarding polyhedra (or even edge-guarding polyhedra) to \SET.

\subsection*{Acknowledgments}
The author wishes to thank the anonymous reviewers for precious suggestions on how to improve the readability of this paper.

\small
\bibliographystyle{abbrv}

\begin{thebibliography}{99}

\bibitem{set}
N.~Alon, D.~Moshkovitz, and S.~Safra.
Algorithmic construction of sets for $k$-restrictions.
{\it ACM Transactions on Algorithms}, vol.~2, pp.~153--177, 2006.

\bibitem{ausiello}
C.~Ausiello, P.~Crescenzi, G.~Gambosi, V.~Kann, A.~Marchetti-Spaccamela, and M.~Protasi.
{\it Complexity and approximation: Combinatorial optimization problems and their approximability properties}. Springer, 2003. 

\bibitem{viglietta4}
N.~Benbernou, E.~D.~Demaine, M.~L.~Demaine, A.~Kurdia, J.~O'Rourke, G.~T.~Toussaint, J.~Urrutia, and G.~Viglietta.
Edge-guarding orthogonal polyhedra.
In {\it Proceedings of the 23rd Canadian Conference on Computational Geometry}, pp.~461--466, 2011.

\bibitem{orthostacks}
T.~Biedl, E.~D.~Demaine, M.~L.~Demaine, A.~Lubiw, M.~H.~Overmars, J.~O'Rourke, S.~Robbins, S.~H.~Whitesides.
Unfolding some classes of orthogonal polyhedra.
In {\it Proceedings of the 10th Canadian Conference on Computational Geometry}, pp.~70--71, 1998.

\bibitem{edgenew}
J.~Cano, C.~D.~T\'{o}th, and J.~Urrutia.
Edge guards for polyhedra in 3-space.
In {\it Proceedings of the 24th Canadian Conference on Computational Geometry}, pp.~155--160, 2012.

\bibitem{wireless}
T.~Christ and M.~Hoffmann.
Wireless localization within orthogonal polyhedra.
In {\it Proceedings of the 23rd Canadian Conference on Computational Geometry}, pp.~467--472, 2011.

\bibitem{quadric}
L.~Dupont, D.~Lazard, S.~Lazard, and S.~Petitjean.
A new algorithm for the robust intersection of two general quadrics.
In {\it Proceedings of the 19th Annual ACM Symposium on Computational Geometry}, pp.~246--255, 2003.

\bibitem{ghosh}
S.~K.~Ghosh.
Approximation algorithms for art gallery problems.
In {\it Proceedings of the Canadian Information Processing Society Congress}, pp.~429--434, 1987.

\bibitem{terrain2}
C.~Iwamoto, J.~Kishi, and K.~Morita.
Lower bound of face guards of polyhedral terrains.
{\it Journal of Information Processing}, vol.~20, pp.~435--437, 2012.

\bibitem{terrain1}
C.~Iwamoto, Y.~Kitagaki, and K.~Morita.
Finding the minimum number of face guards is NP-hard.
{\it IEICE Transactions on Information and Systems}, vol.~E95-D, pp.~2716--2719, 2012.

\bibitem{terrain3}
C.~Iwamoto and T.~Kuranobu.
Improved lower and upper bounds of face guards of polyhedral terrains.
{\it IEICE Transactions on Information and Systems (Japanese Edition)}, vol.~J95-D, pp.~1869--1872, 2012.

\bibitem{art}
J.~O'Rourke.
{\it Art gallery theorems and algorithms}.
Oxford University Press, New York, 1987.

\bibitem{shermer}
T.~Shermer.
Recent results in art galleries.
In {\it Proceedings of the IEEE}, vol.~80, pp.~1384--1399, 1992.

\bibitem{faceguards}
D.~L.~Souvaine, R.~Veroy, and A.~Winslow.
Face guards for art galleries.
In {\it Proceedings of the XIV Spanish Meeting on Computational Geometry}, pp.~39--42, 2011.

\bibitem{urrutia2000}
J.~Urrutia.
Art gallery and illumination problems.
In J.-R.~Sack and J.~Urrutia, editors, {\it Handbook of Computational Geometry}, pp.~973--1027, North-Holland, 2000.

\bibitem{viglietta2}
G.~Viglietta.
Searching polyhedra by rotating half-planes.
{\it International Journal of Computational Geometry and Applications}, vol.~22, pp.~243--275, 2012.

\bibitem{thesis}
G.~Viglietta.
{\it Guarding and searching polyhedra}.
Ph.D.~Thesis, University of Pisa, 2012.

\bibitem{mycccg}
G.~Viglietta.
Face-guarding polyhedra.
In {\it Proceedings of the 25th Canadian Conference on Computational Geometry}, pp.~277--282, 2013.

\end{thebibliography}

\end{document}